\documentclass[10pt]{article}

\usepackage{pdfsync}

\usepackage{amsmath,amsthm,amsfonts,mathrsfs,anysize,fancyhdr,epsfig,mdframed, float,graphicx,umoline,dsfont}

\usepackage{chngcntr}
\usepackage{apptools}

\usepackage{amsmath}                                                        % need for subequations
\usepackage{graphicx}                                                       % need for figures
\usepackage{subfigure}
\usepackage{enumitem}                                                    % need for subfigures
\usepackage{hyperref}
\usepackage{amssymb}                                                        % gives you \mathbb{} font
\usepackage[mathscr]{eucal}                                             % gives you \mathscr font
\usepackage{chngcntr}
\counterwithin{table}{section}
\counterwithin{figure}{section}
\usepackage{cancel}                                                             % gives you the ability to visibly cross out terms in equations}
\usepackage[normalem]{ulem}                                                                 % gives you \sout
\usepackage{pstricks}
\usepackage{rotating}
\usepackage{lscape}
\usepackage[paperwidth=8.5in,paperheight=11in,top=1.25in, bottom=1.25in, left=1.00in, right=1.00in]{geometry}
\usepackage{mathtools}                                                      % need for `show only references'
\mathtoolsset{showonlyrefs=true}                                    % only equations which are labeled AND referenced will be numbered.
\usepackage{fixltx2e,amsmath}                                           % Supposedly, this allows one to use \eqref{} in \caption{}.
\MakeRobust{\eqref}

\usepackage{mathdots}
\usepackage{amsthm}                                                             % need for theorem-proof environment
\allowdisplaybreaks                                                             % allows page breaks for long equations
\theoremstyle{plain}
\newtheorem{theorem}{Theorem}

\newtheorem{lemma}[theorem]{Lemma}                              % [theorem] ==> theorems and lemmas will share a counter
\AtAppendix{\counterwithin{theorem}{section}}
\newtheorem{proposition}[theorem]{Proposition}

\theoremstyle{definition}
\newtheorem{definition}[theorem]{Definition}

\newtheorem{remark}[theorem]{Remark}
\newtheorem{assumption}[theorem]{Assumption}

\def \caratt {{\mathds{1}}}

\providecommand{\keywords}[1]{\textbf{Keywords:} #1}
\newcommand{\norm}[1]{\left\|{#1}\right\|}

\def \phi {{\varphi}}

\usepackage[numbers]{natbib}
\usepackage{hyperref}

\begin{document}

\title{Systemic risk in a mean-field model of interbank lending with self-exciting shocks}

\author{Anastasia Borovykh\thanks{Dipartimento di Matematica, Universit\`a di Bologna, Bologna, Italy.
\textbf{e-mail}: anastasia.borovykh2@unibo.it} \and Andrea Pascucci\thanks{Dipartimento di Matematica, Universit\`a di
Bologna, Bologna, Italy. \textbf{e-mail}: andrea.pascucci@unibo.it}
\and Stefano La Rovere\thanks{NIER Ingegneria, Bologna, Italy. \textbf{e-mail}: s.larovere@niering.it}}

\maketitle

\begin{abstract}
In this paper we consider a mean-field model of interacting diffusions for the monetary reserves in which the reserves are subjected to a self- and cross-exciting shock. This is motivated by the financial acceleration and fire sales observed in the market. We derive a mean-field limit using a weak convergence analysis and find an explicit measure-valued process associated with a large interbanking system. We define systemic risk indicators and derive, using the limiting process, several law of large numbers results and verify these numerically. We conclude that self-exciting shocks increase the systemic risk in the network and their presence in interbank networks should not be ignored.
\end{abstract}
\section{Introduction}
An important financial issue is understanding the risk in financial systems with interacting entities. Many previous research focuses on contagion through interbank lending agreements, however contagion can occur through multiple other channels, e.g. linked balance sheets that may result in fire sales
%. When an institution decides to liquidate a large part of its assets due to e.g. a need for liquidity, this might depress the asset price, affecting other institutions that hold similar assets in their portfolios. The effects of these fire sales are
(see e.g. \citet{capponi15a} and \citet{chen14}) and the so-called financial acceleration. %Whenever an institution receives a negative shock, its counterparties might react by pulling out their funds or imposing stricter lending rules. In both ways this initial shock can result in another shock to the banks asset values.
 In \citet{cont10} the authors argue that one should not
ignore the compounded effect of both correlated market events and default contagion, since it can make the network considerably more vulnerable to default cascades. Motivated by the above mentioned research, we choose to model the effects of the self-exciting fire sales as well as the financial acceleration by including a self- and cross-exciting Hawkes process, as introduced in \citet{hawkes71}, in the dynamics for the monetary reserve of the bank and combine this with the default propagation through interbank lending agreements to study the robustness of the network.

Modelling the financial network can be done using a so-called mean-field model. Here the matrix of interbank borrowing/lending activities is exogenously specified and the dynamics of the banks' monetary reserves depend on stochastic idiosyncratic events and on an interaction term, modelled through an empirical distribution of the system states, which captures the type of interaction with the other nodes in the system. One way of studying these interacting systems is by investigating the behaviour of the system as the number of nodes approaches infinity (i.e. propagation of chaos). In \citet{capponi15} the authors consider an interacting model of the monetary reserve processes where the drift term represents interbank short-time lending and the monetary reserve is additionally subjected to a banking sector indicator which drives additional in-/out-flows of cash. By means of a detailed weak-convergence analysis they conclude that the underlying limit state process has purely diffusive dynamics and the contribution of the banking sector jump process is reflected only in the drift. In \citet{nadtochiy17} the authors use the mean-field approach with an interaction through hitting times in estimating systemic failure. 

The large limit behavior of a system has also been studied in a portfolio setting. By means of a weak converence analysis of \citet{giesecke13} study the behavior of the default intensity in a large portfolio where the intensity is subjected to additional sources of clustering through exposure to a systematic risk factor and a contagion term. The law of large numbers (LLN) result is proven under the assumption that the systematic risk vanishes in the large-portfolio limit. In \citet{giesecke15} the authors extend the previous result for general diffusion dynamics for the systemic risk factor without the vanishing assumption, producing a stochastic PDE for this density in the limit, as opposed to a PDE. In \citet{spiliopoulos14} the LLN result is extended by proving a central limit theorem (CLT) in a similar setting, thus quantifying the fluctuations of the empirical measure (and thereby also the loss from default) around its large portfolio limits. In \citet{bush11} the large portfolio limit for assets following a correlated diffusion is shown to approach a measure whose density satisfies an SPDE, while in \citet{hambly17} a similar result is proven for a stochastic volatility model for the asset price. Finally, \citet{sirignano15} and \citet{sirignano16} use mean-field and large portfolio approximation methods for the analysis of large pools of loans. 

%Since multiple defaults are considered to be rare events, using regular Monte Carlo methods may complicate the computation of these rare event probabilities. A commonly employed variance reduction strategy is the use of importance sampling which samples from a distribution in which the probability of the rare event occurring is higher. The difficulty with implementing importance sampling is computing the change of measure under which the random variables need to be simulated. The authors in \citet{carmona09} propose an Monte Carlo scheme based on the interacting particle system that uses samples from the background Markov chain. In \citet{blanchet12} the authors consider a network consisting of insurers and re-insurers and propose a state-dependent importance sampling Monte Carlo method for computing the probability density of the time of default of a subset of nodes in the network.

The aim of our paper is to investigate the systemic risk in a network when incorporating both self- and cross-exciting shocks as well as interbank lending in the monetary reserve process of the bank. The excitement comes from the effect that past movements in both the asset value of the bank itself as well as that of its neighbors have on the current variations in its asset value. These effects are modelled using a Hawkes process. Self-exciting processes have previously been used in portfolio credit risk computation from a top-down approach, see \citet{ait15}, \citet{errais10} and \citet{cvitanic12}. In this work we model the monetary reserve process of a bank through a mean-field interaction diffusion with an additional Hawkes distributed jump term. We study the behavior of the system as the number of nodes approaches infinity by deriving the weak limit of the empirical measure of this interacting system. %In other words, we start with each monetary reserve of the nodes behaving according to a diffusion with an interacting drift term corresponding to interbank lending activities and subjected to a Hawkes shock and study the behavior of the system as the number of nodes approaches infinity.

In particular our convergence result is based on the analysis of \citet{delattre16}, where the authors show that the intensity of a Hawkes process in the limit of a fully connected network tends to behave as that of a non-homogeneous Poisson process. We show that the underlying limit process for the monetary reserves of the nodes has purely diffusive dynamics and the effect of the Hawkes process is reflected in a \emph{time-dependent} drift coefficient. Then we define several risk indicators and use the weak convergence analysis to derive the law of large numbers approximations to explicitly show the effects of the Hawkes process on the risk in a large interbank network. In the numerical section we then compare the LLN aproximations with the actual values simulated through a Monte-Carlo method and conclude that in a model of interbank networks, the default risk is indeed higher when we incorporate the self- and cross-exciting shocks.

The rest of the paper is structured as follows: in Section \ref{sec2} we define the Hawkes process and give a motivation for incorporating it in the interbank network. In Section \ref{sec3} we introduce the mean-field model for the log-monetary reserve process and study through simulations the effects of incorporating the self-exciting jump intensity and in particular compare it to the independent Poisson intensity. In Section \ref{sec4} we derive the weak convergence of the empirical mean of monetary reserves, explicitly characterize the weak limit measure-valued process and provide several results for extensions of the model. Finally, in Section \ref{sec5} we derive several measures of systemic risk in the network and numerically validate the accuracy of the derived limiting process.

\section{The framework}\label{sec2}
\subsection{Motivation}
A known source of systemic risk in financial networks is the propagation of default due to interbank exposures such as loans, where the failure of a borrowing node to repay its loans, may consequently cause a loss in liquidity of the lenders as well, in this way propagating the default through the network. Besides interbank exposures, another common cause of default propagation are fire sales. If one institution decides to liquidate a large part of its assets, depressing the price, this causes a loss at the institutions holding the same assets, creating a \emph{cross-exciting} spiral across the institutions. Therefore, institutions that do not have mutual counterparty exposure can still suffer financial distress if they have holdings of common assets on their balance sheets. As illustrated by \citet{glasserman15}, the effects from these so-called fire sales can be even greater than the contagion effects due to counterparty exposures.

A \emph{self-exciting} effect present in financial networks is known as financial acceleration and refers to the fact that current variations in the asset side of the balance sheet depend on past variations in the assets themselves. In other words, a shock affecting the banks portfolio can cause creditors to claim their funds back or tighten the credit conditions, in this way causing an additional shock for the bank. %In \citet{battiston12} the authors incorporated this effect by including a jump in the SDE for robustness where the jump at time $t$ depends also on the realization of the robustness at $t-1$. The penalty associated to this jump depends on the change in the robustness. This means that in case a bank is hit by a negative shock its partners react only if this shock is perceived as abnormal given the current market conditions.

As has been mentioned in \citet{cont10}, while interbank lending itself may not be a significant cause of default propagation, it is important to account for both the correlated effects of default contagion through lending agreements as well as exposure to common market events. Here, we choose to model the correlated effects of the fire sales, financial acceleration and the interbank lending structure on both the default propagation as well as overall loss in the network through a Hawkes counting process. The shocks affecting the portfolio of the institution arrive conditional on the infinite history of previous shocks to both the institutions own assets as well as those of the other nodes in the network provided that they share common assets.

\subsection{Hawkes processes}\label{sec21}
Specific types of events that are observed in time do not always arrive in evenly spaced intervals, but can show signs of clustering, e.g. the arrival of trades in an order book, or the contagious default of financial institutions. Therefore, assuming that these events happen independently is not a valid assumption. A Hawkes process (HP), also known as a self-exciting process, has an intensity function whose current value, unlike in the Poisson process, is influenced by past events. In particular, if an arrival causes the conditional intensity to increase, the process is said to be self-exciting, causing a temporal clustering of arrivals. Hawkes processes can be used for modelling credit default events in a portfolio of securities, as has been done in e.g. \citet{errais10} or for modelling asset prices using a mutually exciting jump component to model the contagion of financial shocks over different markets (\citet{ait15}). An overview of other applications of Hawkes processes in finance, in particular in modelling the market microstructure, can be found in e.g. \citet{bacry15}.

Let $(\Omega,{\mathcal F},\mathbb{F},\mathbb{P})$ be a complete filtered probability space where the filtration ${\mathbb{F}}=({\mathcal{F}_t})_{t\geq0}$ satisfies 
the usual condition. Hawkes processes (\citet{hawkes71}) are a class of multi-variate counting processes
$(N_t^1,...,N_t^M)_{t\geq 0}$ characterized by a stochastic intensity vector
$(\lambda_t^1,...,\lambda_t^M)_{t\geq 0}$ which describes the $\mathcal{F}_t$-conditional mean
jump rate per unit of time, where $\mathcal{F}_t$ is the filtration generated by $(N^i)_{1\leq
i\leq M}$ up to time $t$. Consider the set of nodes $I_{M}:= \{1,\dots,M\}$. Define the kernel $g(t)=(g^{i,j}(t),
(i,j)\in I_{M}\times I_{M})$ with $g^{i,j}(t):\mathbb{R}_+\to\mathbb{R}$ and the constant
intensity $\mu=(\mu^i, i\in I_{M})$ with $\mu^i\in\mathbb{R}_+$.
\begin{definition}[Hawkes process]\label{defhawk}
A linear Hawkes process with parameters $(g,\mu)$ is a family of $\mathcal{F}_t$-adapted counting
processes $(N^i_t)_{i\in I_{M},t\geq 0}$ such that:
\begin{enumerate}
\item almost surely for all $i\neq j$, $(N_t^i)_{t\geq 0}$ and $(N_t^j)_{t\geq 0}$ never jump simultaneously,
\item for every $i\in I_{M}$, the compensator $\Lambda_t^i$ of $N_t^i$ has the form $\Lambda_t^i=\int_0^t\lambda_s^ids$, where the intensity process $(\lambda_t^i)_{t\geq 0}$ is given by
\begin{align}\label{eq:hawkesint}
\lambda^i_t = \mu^i + \sum\limits_{j=1}^M\int_{[0,t[}g^{i,j}(t-s)dN^j_s.
\end{align}
\end{enumerate}
\end{definition}
In other words, $g^{i,j}$ denotes the influence of an event of type $j$ on the arrival of $i$:
each previous event $dN^j_s$ raises the jump intensity $(\lambda^i_t)_{i\in I_{M}}$ of its
neighbors through the function $g^{i,j}$. The compensated jump process $N_t-\int_{0}^t\lambda_sds$
is a $\mathcal{F}_t$-local martingale. For $g$ a positive and a decreasing function of time $t$,
the influence of a jump decreases and tends to 0 as time evolves.

Following Proposition 3 in \citet{delattre16}, one can rewrite the Hawkes process in the sense of
Definition \ref{defhawk} as a Poisson-driven SDE with the i.i.d. family of $\mathcal{F}_t$-Poisson
measures $(\pi^i(ds,dz), i\in I_{M})$ with intensity measure $(ds,dz)$:
\begin{align}\label{eq:hawkpoiss}
N_t^i = \int_0^t \int_0^\infty \caratt_{\{z\leq
\mu_t+\sum\limits_{j=1}^M\int_{[0,s[}g^{i,j}(t-s)dN_s^j\}}\pi^i(ds,dz).
\end{align}

Next we state a well-posedness result, based on Theorem 6 in \citet{delattre16}:
\begin{lemma}[Existence and uniqueness]\label{exist} Let $g^{i,j}$ be locally integrable for all $(i,j)\in I_{M} \times I_{M}$;
there exists a pathwise unique Hawkes process $(N_t^i)_{i\in I_{M},t\geq 0}$, such that
$\sum\limits_{i=1}^M\mathbb{E}[N_t^i]<\infty$ for all $t\geq 0$.
\end{lemma}

By introducing the pair $\{t_k, n_k\}_{k=1}^{K_t}$, where $t_k$ denotes the time of event $k$,
$n_k\in I_{M}$ is the event type and $K_t=\sum\limits_{i=1}^M N_t^i$ is total number of event
arrivals up to time $t$, we can rewrite the intensity as
\begin{align*}
\lambda^i_t = \mu^i + \sum\limits_{k=1}^{K_t}g^{i,n_k}(t-t_k), \qquad i \in I_M.
\end{align*}
A common choice for $g^{i,j}(t)$ is an exponential decay function defined as
\begin{align}\label{eq:expjump}
g^{i,j}(t) = \alpha^{i,j}e^{-\beta^i t},
\end{align}
so that $\lambda^i_t$ jumps by $\alpha^{i,j}$ when a shock in $j$ occurs, and then decays back
towards the mean level $\mu^i$ at speed $\beta^i$. Note that this function satisfies the local
integrability property, i.e. $g^{i,j}\in L^1_{\text{loc}}(\mathbb{R}_+)$. If $g^{i,j}$ is
exponential then the couple $(N_t,\lambda_t)$ is a Markov process \cite{bacry15}. The simulation of a Hawkes
process can be done using what is known as Ogata's modified thinning algorithm, see for more details \citet{ogata81} and \citet{daley07}.

If the Hawkes process $(N_t^i)_{i\in I_{M},t\geq 0}$ satisfies certain conditions, we have
the following stationarity result (see \citet{bremaud96} and \citet{bacry14} for details), which
will come in useful in the further sections.
\begin{proposition} \label{nonexpl}
Suppose that the matrix $\Phi$ with entries $\int_0^\infty |g^{i,j}(t)|dt$ has a spectral radius strictly less than one.  %In particular, for the exponential excitation, we have $\Phi_{i,j}=\alpha^{i,j}/\beta^i$.
Then there exists a unique multi-variate Hawkes process $(N^i_t)_{t\geq 0}$ for $i\in  I_{M}$ with
stationary increments and the associated intensity as in \eqref{eq:hawkesint} is a stationary
process. Moreover we have %$\mathbb{E}[\lambda_t^i]<\infty$ and
$\mathbb{E}[|\lambda_t|^2]< \infty$.
\end{proposition}
Furthermore, we remark here that a multi-dimensional Hawkes process with stationary increments is uniquely defined by its first- and second-order statistics (\citet{bacry14}).

\section{The mean-field model}\label{sec3}
In this section we define the mean-field model for the log-monetary reserves of each of the nodes in the model. The interaction between the nodes is defined through the drift term and additionally we consider the reserve process to be subjected to a self- and cross-exciting Hawkes distributed shock.
\subsection{Definition}
Define $\mathcal{F}_t=\sigma((W_s^i, N_s^i),0\leq s\leq t, i\in\mathbb{N})$. Assume that, for $i \in  I_{M}$ the log-monetary reserves of the $i$-th bank satisfies the following stochastic differential equation (SDE) %(see also \citet{ait15}):
\begin{align}
dX^i_t = \frac{a^i}{M}\sum\limits_{k = 1}^{M}(X^k_t-X^i_t)dt + \sigma^i d W_t^i + c^i dN^i_t,
\end{align}
with $X_0^i\in\mathbb{R}_+$ the initial reserves for each bank and where $a^i\geq 0$,
$\sigma^i\geq 0$ and $c^i:=\hat c^i/M<0$ are constants for each $i\in I_{M}$.
The process $W(t) = \{ W_t^i\}_{i=1}^M$ is a $M$-dimensional Brownian motion, 
and $N_t = \{N_t^i\}_{i=1}^M$ is the vector of Hawkes processes with self-exciting intensity $\lambda^i_t$ as defined in \ref{sec21}. With the drift term defined in this way, if bank $k$ has more (less) log-monetary reserves than bank $i$, i.e. $X_t^k>X_t^i$ ($X_t^k<X_t^i$), bank $k$ is assumed to lend (borrow) a proportion of the surplus to (deficit from) bank $i$, with proportionality factor $a^i/M$. A jump in the Hawkes process $i$ affects the corresponding $X_t^i$ through the proportionality factor $c^i$ and increases the intensity $\lambda_t^j$ for $j\in I_M$ if $g^{i,j}(t)\neq 0$. In this way the jump activity varies over time resulting in a clustering of the arrival of the jumps and the shocks propagate through the network in a contagious manner through the contagion function $g^{i,j}(t)$. We thus interpret the jump term $c^i dN_t^i$ as a self- and cross-exciting \emph{negative} effect that occurs due to financial acceleration and fire sales, resulting in a decrease in a banks monetary reserve. In \citet{capponi15} the authors considered a similar mean-field model for the monetary reserves but assumed the jumps to occur at independent Poisson distributed random times. However, not accounting for the clustering effect of the jumps might cause a significant underestimation of the systemic risk present in the network. We define a default level $D \leq 0$ and say that bank $i$ is in a default state at time $T$ if its log-monetary reserve reached the level $D$ at time $T$. We remark that in our model even if bank $i$ has defaulted, i.e. its monetary reserve reaches a negative level, it continues to participate in the interbank activities borrowing from the counterparties until it e.g. reaches a positive reserve level again. In other words, the level of monetary reserves takes in values in $\mathbb{R}$. %\blu{(This can be resolved by adding the $X_t$ term I guess, but then it adds problems with regard to parameter choices and the zero level never being reached.)}.
We will work in the following setting:
\begin{assumption}[Parameters]\label{ass1}
We collect the parameters associated with the dynamics of the $i$-th monetary reserve process
$i\in I_{M}$ as
\begin{align}
p^i:=(a^i, \sigma^i, c^i)\in (\mathbb{R}_+\times\mathbb{R}_+\times\mathbb{R}_-).
\end{align}
We denote by $\delta_{x}$ the Dirac-delta measure centered at $x$ and we set
$$q^M=\frac{1}{M}\sum\limits_{i=1}^M\delta_{p_i},\qquad
 \phi_0^M=\frac{1}{M}\sum\limits_{i=1}^M\delta_{X_0^i}.$$
We assume $\lim\limits_{M\rightarrow\infty}q^M = \delta_{p^*}$, i.e. $p^i\rightarrow
p^*:=(a,\sigma,c)$ as $i\rightarrow\infty$ and $\lim\limits_{M\rightarrow\infty}\phi_0^M =
\delta_{x}$, i.e. $X_0^i\rightarrow x$ as $i\rightarrow\infty$. We take the exponential decay
function for the contagion
 \begin{align}
g^{i,j}(t-s) = \frac{1}{M}g(t-s):=\frac{1}{M}\alpha e^{-\beta (t-s)},
\end{align}
 which is a locally square-integrable function with $\alpha,\beta\in\mathbb{R}_+$. Finally, the parameters are assumed to all be bounded by a constant $C_p$.
\end{assumption}
We remark here that the results developed in this paper hold also for more general distributions,
i.e. $\lim\limits_{M\rightarrow\infty}q^M=q$ and $\lim\limits_{M\rightarrow\infty}\phi_0^M =
\phi_0$, but for simplicity of the results we assume the parameter vector converges to a constant
vector.

Defining the reserve average as
\begin{align}
\bar X_t = \frac{1}{M}\sum\limits_{i=1}^MX^i_t,
\end{align}
we can rewrite the SDE as a mean-field interaction SDE
\begin{align}\label{eq:meanfieldsde}
dX^i_t = a^i(\bar X_t-X^i_t)dt + \sigma^i dW_t^i + c^idN^i_t.
\end{align}
From \eqref{eq:meanfieldsde} we see that the processes $(X^i_t)$ are mean-reverting to their ensemble average $(\bar X_t)$ at rate $a^i$.

\begin{lemma} There exists a unique solution $(X_t^1,...,X_t^M)$ to the system of SDEs given by \eqref{eq:meanfieldsde} for $i\in  I_{M}$.
\end{lemma}
\begin{proof}
%\blu{(Check this, check that no problems with the intensity of the jump changing; that uniqueness still holds and also in case of no $X_t^{1/2}$ in the BM part that uniqueness of SDE without jumps holds)}.
The proof is similar to Theorem 9.1 in \citet{ikeda81}. Define $Y_t^i$ to be the solution of
the
SDE \eqref{eq:meanfieldsde} without jumps.
By Example 2 in \citet{cox95}, we know that the SDE has a unique strong solution
$(Y_t^1,...,Y_t^M)$. By definition of a Hawkes process we have that $N^1,...,N^M$ never jump
simultaneously: this implies the existence of an increasing sequence of jump times $(\tau_n)_{n\in
\mathbb{N}}$ such that %so that $0<\tau_1<...<\tau_n<...<+\infty$ and
$\lim\limits_{n\rightarrow\infty}\tau_n=+\infty$. Then we can define
\begin{align}\label{eq:defunique}
X_t^{(i,1)}:=\begin{cases} Y_t^i, \qquad 0\leq t< \tau_1,\\ Y_{\tau_1-}^i+\caratt_{k=i}c^i, \qquad
t=\tau_1, \textnormal{ if there is a jump in } N^k.
\end{cases}
\end{align}
From Lemma \ref{exist}, we know that there exists a unique Hawkes process $(N_t^i)_{t\geq 0}$ for
$i\in I_{M}$, thus we can say that $X_t^{(i,1)}$ is the unique solution to \eqref{eq:meanfieldsde}
for $t\in[0,\tau_1]$. Then we define $\bar X_t^{(i,2)}$ on $t\in [0,\tau_2-\tau_1]$ similar to
\eqref{eq:defunique} using as initial state $\bar X_0^i:=X_{\tau_1}^{(i,2)}$ and driving factors
$\bar W_t^i:=W_{t+\tau_{1}}^i-W_{\tau_{1}}^i$ and $\bar N_t^i:= N_{t+\tau_{1}}^i-N_{\tau_{1}}^i$.
Then we set
\begin{align}
X_t^{i}:=\begin{cases} X_t^{i,1}, \qquad 0\leq t < \tau_1,\\ \bar X_{t-\tau_1}^{(i,2)} \qquad
\tau_1\leq t\leq \tau_2.
\end{cases}
\end{align}
So that $X_t^i$, $t\in [0,\tau_2]$ is the unique solution to \eqref{eq:meanfieldsde}. Iterating the above process, we have that $X_t^i$ is determined uniquely on the time interval $[0,\tau_n]$ for each $n\in\mathbb{N}$.
\end{proof}

\subsection{Simulation}
Consider, for the sake of illustration, the following SDE
\begin{align}
dX^i_t = a(\bar X_t-X^i_t)dt + \sigma d\tilde W_t^i + c dN^i_t,
\end{align}
with $\tilde W_t^i := \rho W_t^0 + \sqrt{1-\rho^2}W_t^i$, where $W_t^i$, $i=0,...,M$ are
independent Brownian motions and $W_t^0$ represents common noise (similar to the setting in \citet{carmona13}). We keep the parameters of the constant intensity and the excitation function
$g^{i,j}=\alpha^{i,j}e^{-\beta^i t}$ fixed at $\mu^i = 10/M$, $\beta^{i}=2/M$ and
$\alpha^{i,j}=2/M$ and the initial reserve value is set at $X_0=0$.

\begin{table}[H]
\caption{Parameters corresponding to the various scenarios of the realizations of $(X_t^i, i=1,...,10)$.}\label{tabpar}
\begin{center}
\begin{tabular}{c||c|c|c|c}
Scenario&$a$&$\sigma$&$c$&$\rho$\\\hline\hline\vspace{-0.1cm}
No lending, independent BMs &0&1&0&0.2\\\vspace{-0.1cm}
Lending, independent BMs&10&1&0&0\\\vspace{-0.1cm}
No lending, correlated BMs&0&1&0&0.2\\\vspace{-0.1cm}
Lending and correlated BMs&10&1&0&0.2\\\vspace{-0.1cm}
Lending, correlated BMs and Poisson jumps & 10&1&0.2&0.2\\\vspace{-0.1cm}
Lending, correlated BMs and Hawkes jumps &10&1&0.2&0.2
\end{tabular}
\end{center}
\end{table}
\vspace{-1cm}

\begin{figure}[H]
\begin{center}
  \caption{One realization of $(X_t^i, i=1,...,10)$, $t=1,...100$ with no lending and independent Brownian motions (left), lending and correlated Brownian motions (center) and lending, correlated Brownian motions and the Hawkes distributed jump (with the jump times shown as dots) (right).}\label{fig0}
    \includegraphics[width=0.32\textwidth]{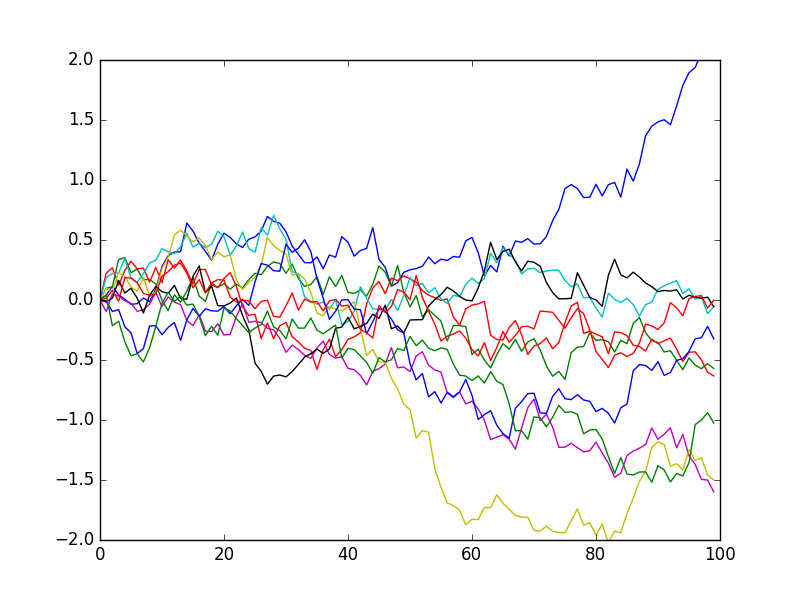}
    \includegraphics[width=0.32\textwidth]{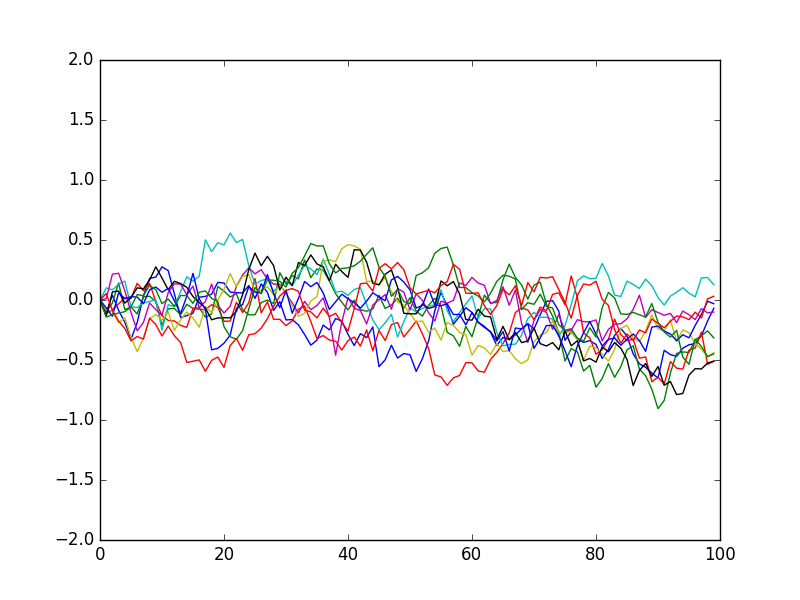}
    \includegraphics[width=0.32\textwidth]{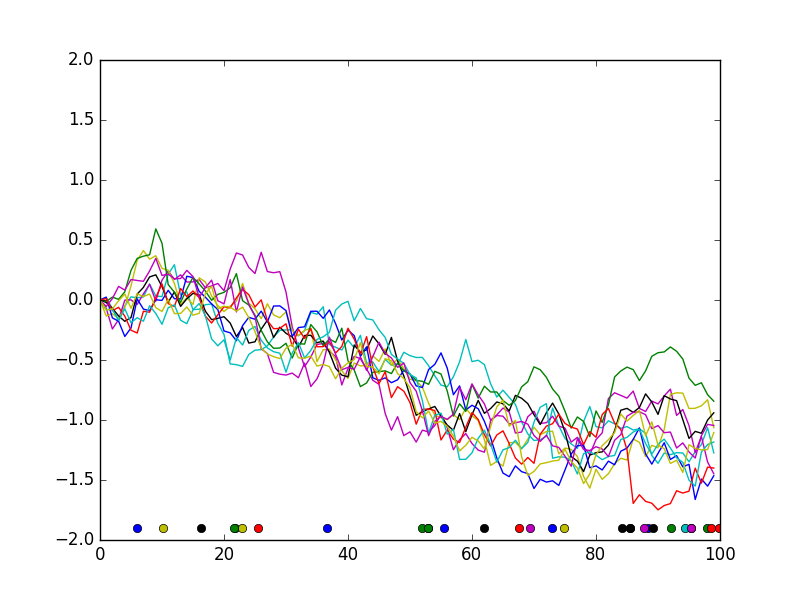}
    \end{center}
\end{figure}

We consider several scenarios of the monetary reserve process denoted in Table \ref{tabpar}. In Figure \ref{fig0} we see that the trajectories generated by the correlated Brownian motions with lending are more grouped than the ones generated by independent Brownian motions without lending. The Hawkes shock, as expected, causes much more trajectories to reach the default level, due to it being an additional source of default propagation.

Consider the default level $D=-0.7$. In Figure \ref{fig1} we show the distributions of the number
of defaults defined as $P\left(\sum\limits_{i=1}^M\left(\min\limits_{0\leq t\leq T}X_t^i\leq
D\right)=n\right)$, for the independent Brownian motion case, the dependent case and the cases
including a Poisson process and a Hawkes process. We observe that the mean-field interbank lending
causes most of the probability mass to be set around zero defaults, as opposed to the no lending
case when the density function is centered at 5 defaults. However, the lending component also adds
a non-negligible probability of all nodes defaulting at once. The correlation between the Brownian
motions affects the loss distribution only slightly. As expected, adding the self-exciting and
clustering Hawkes process increases the tail-risk even more so that the probability of all nodes
reaching a default state rises significantly.

\begin{figure}[H]
\begin{center}
  \caption{The distribution of the number of defautls in several different scenarios, as explained in Table \ref{tabpar}. The parameters in the Monte Carlo simulated based on a discretized Euler-Maruyama scheme are $M=10$, $T=1$, 10000 simulations and 100 time steps.}\label{fig1}
    \includegraphics[width=0.5\textwidth]{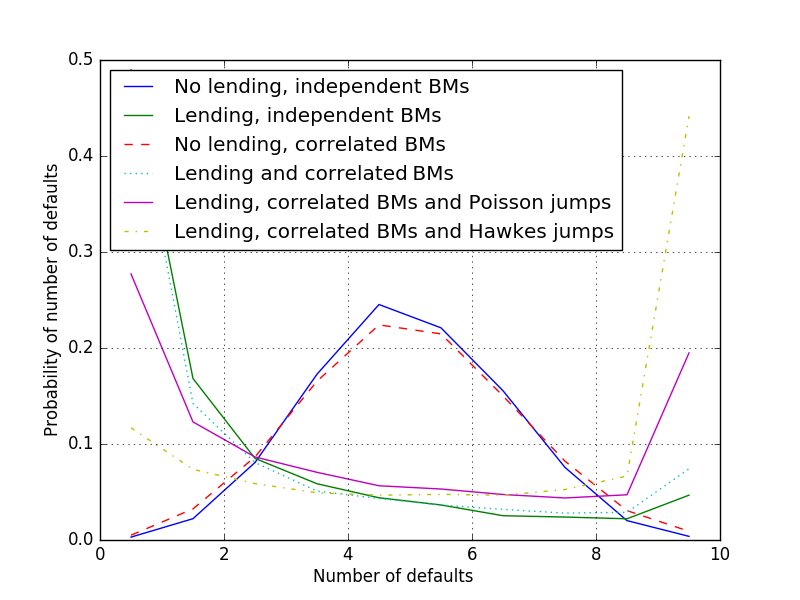}
    \end{center}
\end{figure}

\subsection{Dependency}
As we have already seen in Figure \ref{fig1}, the Hawkes process increases the probability of
multiple defaults occuring at the same time more than an independent Poisson process does. It is
therefore of interest to study the dependence structure between the nodes in more detail. As is
standard in multi-variate statistics, see \citet{poon03}, a tool for assessing the (not necessarily
linear) dependency between variables is the measure $p(q)$ given by
\begin{align}
p(q) = P\left(X^i > F_{X^i}^{-1}(q)|X^j>F_{X^j}^{-1}(q)\right),\qquad i,j\in I_M,
\end{align}
the probability of one of the variables $X^i$ being above the $q$th quantile of its marginal distribution $F_{X^i}$ conditional on the other variable $X^j$ being above its $q$th quantile. To remove the influence of marginal aspects it is typical to transform the data to a common marginal distribution, with e.g. a transformation to unit Fr\'echet marginals (for details we refer to the methodology in \citet{poon03}). In the presence of a dependence between two nodes in our model, the probability of default of one firm conditional of the default of the other will be significant. When computing the systemic risk present in interconnected financial networks, quantifying this dependence is clearly of key importance. Note that in our model we have two key dependencies present:
\begin{itemize}
\item Dependence through the drift term: a high $X_t^1$ results in a change in $X_{s}^1$ and $X_{s}^2$ for $s>t$ due to the interbank loans.
\item Dependence through the Hawkes process: if $\Delta X_t^1<<0$ represents the occurence of a jump at time $t$, then the likelihood of $\Delta X_{s}^1<<0$ and $\Delta X_s^2<<0$ for $s>t$ increases. We remark that the likelihood of seeing the shock decreases with a larger $s$ due to the mean-reverting excitation function $g^{i,j}$ $i,j \in \{1,2\}$.
\end{itemize}
Figure \ref{figscat} shows the scatter plots for both an independent Poisson jump and a Hawkes jump. Already here we see that the Hawkes jump seems to reflect a more strong dependency in the tails.
In Figure \ref{fig3} we plot the measure $p(q)$ (for the left tail) compared to the $1-q$ function representing independence, for several different parameter sets. We see that the Hawkes process shows significantly more dependence between the two nodes for all quantiles compared to the Poisson process. In particular, we note that having only a jump term in the monetary reserve process results in a significant tail probability, where the tail probability of the Hawkes process is considerably higher than that of the Poisson process. This is to be expected since the self-exciting nature of the jumps causes the extreme events in one node to influence extreme events in the other node. Furthermore, incorporating the independent Brownian motion seems to reduce the tail risk almost to zero, while adding the interbank loans in turn causes a slight increase in the tail risk, due to the additional source of default propagation.

\begin{figure}[H]
\begin{center}
  \caption{Scatter plots of $X_t^1$ and $X_t^2$ ($M=2$) showcasing the dependence structure between the nodes in the presence of a Poisson jump (left) and a Hawkes jump (right).}\label{figscat}
  \includegraphics[width=0.3\textwidth]{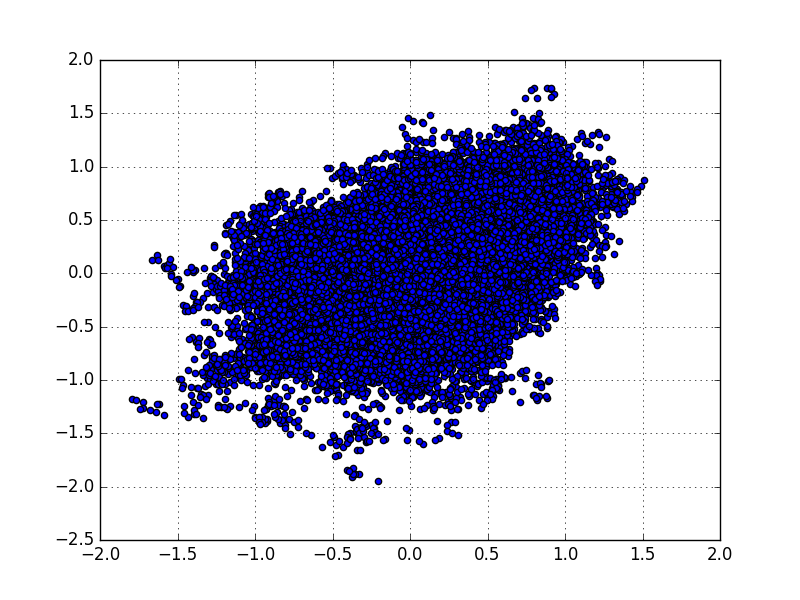}
    \includegraphics[width=0.3\textwidth]{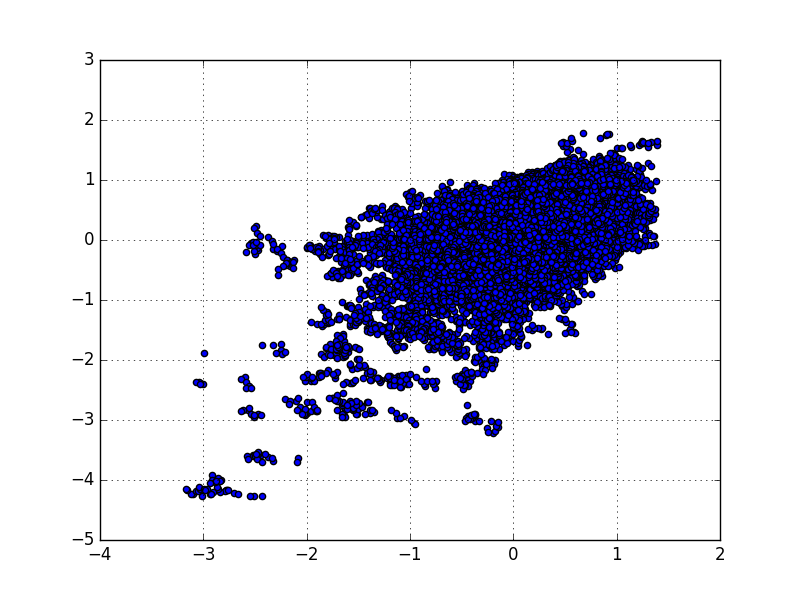}
    \end{center}
\end{figure}
\vspace{-1cm}
\begin{figure}[H]
\begin{center}
  \caption{The measure $p(q)$ quantifying the dependence of $X_t^1$ and $X_t^2$ ($M=2$) with Poisson and Hawkes jumps for the case of no Brownian motion, no interbank lending but only jumps (left, $\sigma = 0$, $a=0$ and $c = -1$), Brownian motion, no lending and jumps (center, $\sigma = 0.1$, $a=0$ and $c=-1$) and Brownian motion, lending and jumps (right, $\sigma = 0.1$, $a=0.5$ and $c=-1$). The other parameters in the Monte Carlo simulation based on a Euler-Maruyama scheme are $T=1$, 500 simulations, 100 time steps, $X_0^i = 0$, $\rho = 0$, with $\mu^i = 0.1$, $\beta^{i}=1.2$, $\alpha^{i,j}=1.2$. }\label{fig3}
  \includegraphics[width=0.3\textwidth]{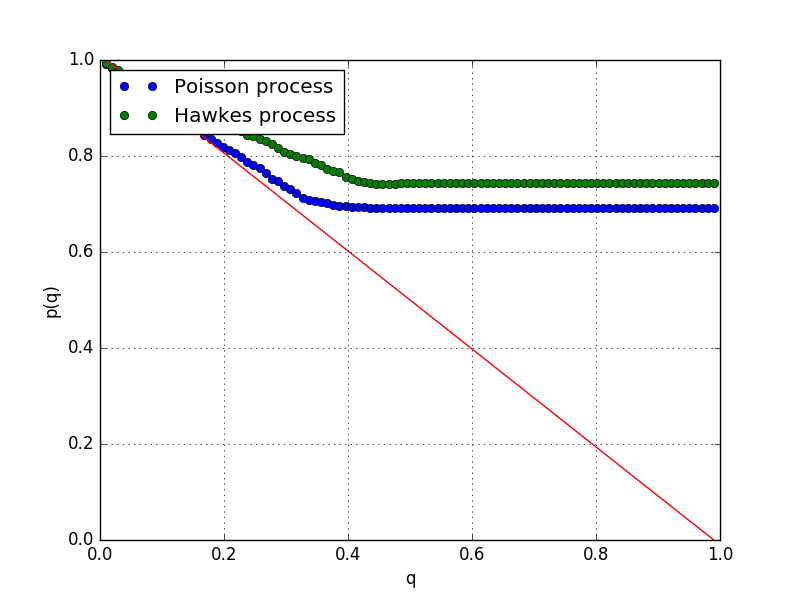}
    \includegraphics[width=0.3\textwidth]{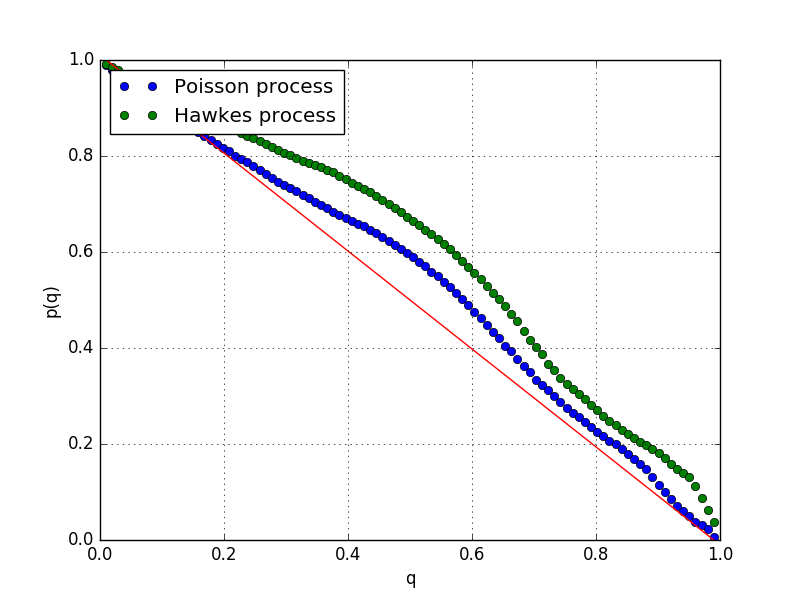}
    \includegraphics[width=0.3\textwidth]{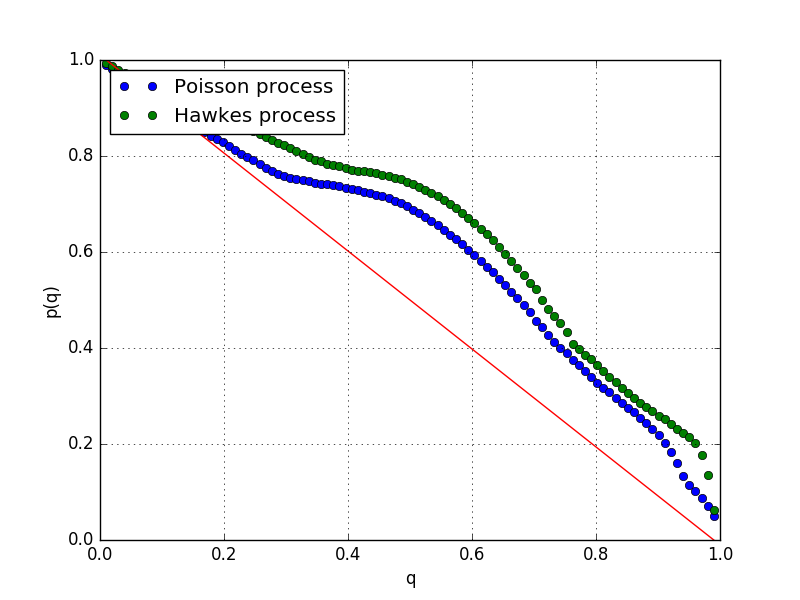}
    \end{center}
\end{figure}

\section{Mean-field limit}\label{sec4}
We derive theoretical mean-field limits for the monetary reserve process with a Hawkes jump term to show the effects of considering this additional type of contagion on the total losses in the network in the case of the number of nodes tending to infinity. Our derivations are based on \citet{carmona13} and \citet{capponi15}. In other words, we want to understand the behavior of the distribution of  the process $X_t=(X_t^i)$, $i\in I_M$ as in \eqref{eq:meanfieldsde} when $M\rightarrow\infty$. % where the Brownian motions are independent \blu{(CHECK with Capponi, but I think yes)}.
Let the vector $(p^i, X_t^i)$ take on values in the space $\mathcal{O}:=(\mathbb{R}_+\times\mathbb{R}_+\times\mathbb{R}_-)\times\mathbb{R}$. Define the sequence of empirical measures as % do i need to use also the intensity in the density to make it a Markov process?
\begin{align}\label{eq:nuemp}
\nu_t^M:=\frac{1}{M}\sum\limits_{i=1}^M\delta_{(p^i,X_t^i)}, \qquad t\geq 0,
\end{align}
on the Borel space $\mathcal{B}(\mathcal{O})$. In other words we keep track of the empirical
distribution of the type, intensity and monetary reserve for all nodes. Let $S=\mathcal{P}(\mathcal{O})$; the
collection of Borel probability measures on $\mathcal{O}$. Then $(\nu^M_t)_{t\geq 0}$ is an
element of the Skorokhod space $D_S[0,\infty)$, %this is ok right [0,\infty)? because it means that \nu_t^M are RCLL processes on [0,\infty) takng values in S, where the [0,\infty) relates to the time, since \nu_t^M is a fucntion of time? my question here is to be sure that it is not in space, because in space the inputs are in \mathbb{R} since X_t can also be negative! unlike Capponi. In any case I think the measure itself is not signed! 
 i.e. it can be viewed as an $S$-valued right-continuous, left-hand
limited stochastic process. For any smooth function $f(p,x)\in C^{\infty}(\mathcal{O})$ defined
for $(p,x)\in \mathcal{O}$  define the integral w.r.t. the measure $\nu$ by
%\blu{(check if this definition is correct... what exactly do we integrate over if we assume all constants? anyways we can easily assume also an initial distribution over $x_0$ so then it would make sense)} Yes is correct since the nu part is a probability so it makes sence
\begin{align}\label{eq:nu1f}
\nu(f):=\int_\mathcal{O}f(p,x)\nu(dp\times dx),
\end{align}
so that
\begin{align}\label{eq:nuf}
\nu_t^M(f) = \frac{1}{M}\sum\limits_{i=1}^Mf(p^i,X_t^i), \qquad t\geq 0.
\end{align}
Then we have $\bar X_t =\nu_t^M(I)$ where $I(x)= x$.

We want to understand the dynamics for $\nu_t^M$  for large $M$. In deriving the limit of the
process $\nu_t^M$ for $M\rightarrow \infty$ we use an argument similar to \citet{capponi15}
and \citet{giesecke13}. In particular, the focus here is on identifying the limiting dynamics, using the result of \citet{delattre16} on the behavior of Hawkes processes in a large system. We identify the limit through the generator of the limiting martingale problem in Section
\ref{sec41}, and subsequently in Section \ref{sec43} we identify the limit process.

\subsection{Weak convergence}\label{sec41}
We want to use the martingale problem to show that $\nu^M_t$ converges to a limiting process. For notational convenience we will write $f(X_t^i):=f(p^i,X_t^i)$. By the definition of a Hawkes process we have that for all $i\neq j$, $(N_t^i)_{t\geq 0}$ and $(N_t^j)_{t\geq 0}$ never jump simultaneously and a jump in one of the processes $dN_t^i$ results in only $X_t^i$ having a jump of size $c^i$. Therefore, applying It\^ o's formula gives
\begin{align}
df(X_t^i) =&
a^i\partial_xf(X_t^i)[\nu_t^M(I)-X_t^i]dt+\frac{1}{2}(\sigma^i)^2\partial_{xx}f(X_t^i)dt+\sigma^i\partial_xf(X_t^i)dW_t^i\\
&+(f(X_{t-}^i+c^i)-f(X_{t-}^i))dN_t^i,
\end{align}
Then we have, using the definition of $\nu_t^M$ in \eqref{eq:nuf},
\begin{align}\label{eq:nuito}
\nu_t^M(f) =&
\nu_0^M(f)+\int_0^t\nu_s^M(\mathcal{L}^1f)\nu_s^M(I)ds-\int_0^t\nu_s^M(\mathcal{L}^2f)ds+\frac{1}{M}\sum\limits_{i=1}^M\int_0^t\sigma^i\partial_xf(X_s^i)dW_s^i\\
&+\int_0^t\nu_s^M(\mathcal{L}^3f)ds+\frac{1}{M}\sum\limits_{i=1}^M\int_0^t\left(f(X_{t-}^i+c^i)-f(X_{t-}^i)\right)dN_s^i,
\end{align}
where we have defined the operators $\mathcal{L}^*$ acting on $f(p^i,X_t^i)$ as
\begin{align}
\mathcal{L}^1f(p,x):=a\partial_xf(p,x), \qquad \mathcal{L}^2f(p,x):=ax\partial_xf(p,x), \qquad
\mathcal{L}^3f(p,x)=\frac{1}{2}\sigma^2\partial_{xx}f(p,x),
\end{align}
so that
\begin{align}
\nu_t^M(\mathcal{L}^1f) = \frac{1}{M}\sum\limits_{i=1}^M
a^i\partial_xf(p^i,X_t^i),\;\;\nu_t^M(\mathcal{L}^2f) = \frac{1}{M}\sum\limits_{i=1}^M
a^iX_t^i\partial_xf(p^i,X_t^i),\;\;\nu_t^M(\mathcal{L}^3f) = \frac{1}{M}\sum\limits_{i=1}^M
\frac{1}{2}(\sigma^i)^2\partial_{xx}f(p^i,X_t^i).
\end{align}
Define for any smooth function $\phi\in C^\infty(\mathbb{R}^N)$ with $N\in\mathbb{N}$ and Borel measure $\nu\in S$
\begin{align}\label{eq:phi}
\Phi(\nu) = \phi(\nu(\bold{f})),
\end{align}
with $\bold{f}=(f_1,...,f_N)$ for $f_n\in C^\infty(\mathcal{O})$, $n=1,...,N$ and
$\nu(\bold{f}):=(\nu(f_1),...,\nu(f_N))\in\mathbb{R}^N$. Let $\mathcal{S}$ be the collection
of bounded measurable functions $\Phi$ on $S$. Then $\mathcal{S}$ separates $S$ and it thus
suffices to show convergence of the martingale problem for those functions. Then, by
applying It\^o's formula to $\phi(\nu_t^M(\bold{f}))$ and using the fact that $d\tilde
N_t^i:=dN_t^i-\lambda_t^idt$ and $dW^i_t$ are martingales and $X_{t-}$ and $\lambda_{t-}$ are
predictable, we find for $0\leq t<u$
\begin{align}
\Phi(\nu_u^M)=\Phi(\nu_t^M)+\int_t^u\left(\mathcal{C}_s^M+\mathcal{D}_s^M+\mathcal{J}_s^M\right)ds+\mathcal{M}_u-\mathcal{M}_t,
\end{align}
where $(\mathcal{M}_t)_{t\geq 0}$ is an initial mean-zero martingale and
\begin{align}
&\mathcal{C}_t^M:=\sum\limits_{n=1}^N\frac{\partial\phi(\nu_t^M(\bold{f}))}{\partial
f_n}\left(\nu_t^M(\mathcal{L}^1f_n)\nu_t^M(I)-\nu_t^M(\mathcal{L}^2f_n)+\nu_t^M(\mathcal{L}^3f_n)\right),\\
&\mathcal{D}_t^M:=\frac{1}{2M^2}\sum\limits_{n,l=1}^N\frac{\partial^2\phi(\nu_t^M(\bold{f}))}{\partial
f_n\partial f_l}\sum\limits_{i=1}^M\left((\sigma^i)^2\frac{\partial
f_n(X_t^i)}{\partial x}\frac{\partial f_l(X_t^i)}{\partial x}\right)\\
&\mathcal{J}_t^M:=\sum\limits_{i=1}^M\left[\phi(\nu_t^M(\bold{f})+J_t^{M,i}(\bold{f}))-\phi(\nu_t^M(\bold{f}))\right]\lambda_t^i,
\end{align}
where $J_t^{M,i}(\bold{f})=(J_t^{M,i}(f_1),...,J_t^{M,i}(f_N))$ and
\begin{align}
J_t^{M,i}(f):=\frac{1}{M}(f(X_{s-}^i+c^i)-f(X_s^i)).
\end{align}
We will need the following result given in Theorem 8 in \citet{delattre16}:
\begin{theorem}[Propagation of chaos result for the Hawkes process]\label{lemmaint}
Consider the Hawkes process in the sense of \eqref{eq:hawkpoiss}. For each $M\geq 1$ consider the
complete graph with nodes $I_{M}$. Let $g:[0,\infty)\rightarrow\mathbb{R}$ be a locally square-integrable function and set $g^{i,j}= M^{-1}g$ for all $i,j\in I_{M}$. Define the limit equation
\begin{align}\label{eq:barn}
\bar N_t = \int_0^t\int_0^\infty \caratt_{\{z\leq \left(\mu_t+\int_0^sg(t-s)d\mathbb{E}[\bar
N_s]\right)\}}\pi(ds,dz),
\end{align}
where $\pi(ds,dz)$ is a Poisson measure on $[0,\infty)\times [0,\infty)$ with intensity measure
$dsdz$. Then we have $d\mathbb{E}[\bar N_t] = \bar \lambda_tdt$ and
\begin{align}\label{eq:delambda}
\bar\lambda_t := \mu + \int_0^tg(t-s)d\mathbb{E}[\bar N_s].
\end{align}
In other words $\bar N=(\bar N_t)_{t\geq 0}$ is an inhomogeneous Poisson process with intensity $\bar \lambda_t$. Let $\bar N_t^i$ be an i.i.d. family of solutions to \eqref{eq:barn} for $i\in I_M$. Define $\Delta_M^i(t)=\int_0^t|d(\bar N_u^i-N_u^i)|$ and $\delta_M(t)=\mathbb{E}[\Delta^i_M(t)]$. Note that this $\delta_M(t)$ does not depend on $i$ due to exchangeability of both $\bar N_t^i$ and $N_t^i$. Then,
\begin{align}
\delta_M(t)=\int_0^t\mathbb{E}\left[\left|\bar\lambda_t-\lambda^i_t\right|\right]ds,
\end{align}
and for $t\in [0,T]$ we have
\begin{align}
\lim\limits_{M\rightarrow\infty}\delta_M(t)=0.
\end{align}
\end{theorem}
%\begin{proof}
%The statement follows from Theorem 8 in \citet{delattre16}.
%%is a simple extension of the proof of Theorem 8 in \citet{delattre16}. For this, note that
%%\begin{align}
%%\delta_M(t)\leq& \int_0^t\mathbb{E}\left[\left|\int_0^s g(s-u)d\bar\lambda_u-M^{-1}\sum\limits_{i=1}^M\int_0^s g(s-u)d\bar N_u^i\right|\right]ds\\
%%&+\int_0^t\mathbb{E}\left[\left|M^{-1}\sum\limits_{i=1}^M\int_0^s g^{i,j}(s-u)d[\bar N_u^i-N_u^i]\right|\right]ds\\
%%&+\int_0^t\mathbb{E}\left[\int_0^s\left| M^{-1}g(s-u)-g^{i,j}(s-u)\right|d \bar N_u^i\right]ds\\
%%=&A+B+C.
%%\end{align}
%%Using exchangeability due to $\bar N_t^i$ being i.i.d. copies of $\bar N_t$
%%Now we remark that using Lemma 22 in \citet{delattre16}
%%\begin{align}
%%C\leq \sum\limits_{i=1}^M\int_0^t(g^{i,j}(t-s)-M^{-1}g(t-s))m_sds,
%%\end{align}
%%so that if $M\rightarrow\infty$ we have that $C\rightarrow 0$ using $g^{i,j}\rightarrow M^{-1}g$ for $i\rightarrow\infty$. Then, applying the bounds for $A$ and $B$ and Lemma 23.1 from \citet{delattre16} the statement follows.
%\end{proof}
In other words, when all nodes interact in the same way in the limit of the number of nodes going
to infinity, the Hawkes process reduces to an inhomogeoneous Poisson process and we have for any
$i\in I_{M}$ the following limit
\begin{align}\label{eq:delattre}
\lim\limits_{M\rightarrow\infty}\mathbb{E}\left[\int_t^u|\lambda_s^i-\bar\lambda_s|ds\right]=0.
\end{align}

The task is now to find the generator of the limiting martingale problem which we will use to determine the process governing the dynamics of the monetary reserves in the limit, see e.g. Theorem 8.2 Chapter 4 of \citet{ethier86}. For this we will use \eqref{eq:delattre} and define a Taylor-based simplification of $\mathcal{J}_t^M$ as %\blu{(check if reducement to one sum due to indicator is correct)->yes is ok} %see also B.4 in Bo - Bilateral CVA
\begin{align}
\mathcal{\tilde J}_t^M:=\sum\limits_{n=1}^N\frac{\partial\phi(\nu_t^M(\bold{f}))}{\partial x_n}
\left[\frac{1}{M}\sum\limits_{i=1}^M\bar\lambda_t\frac{\partial f_n(X_t^i)}{\partial x}c^i
\right].
\end{align}
Using the triangle inequality we have
\begin{align}
\mathbb{E}&\left[\int_t^u|\mathcal{J}_s^{M}-\mathcal{\tilde J}_s^M|ds\right] \label{eq:ineqs}\\
&\leq
\mathbb{E}\left[\int_t^u\left|\sum\limits_{i=1}^M\left[\phi(\nu_s^M(\bold{f})+J_s^{M,i}(\bold{f}))-\phi(\nu_t^M(\bold{f}))\right]\lambda_s^i-\sum\limits_{i=1}^M\left[\sum\limits_{n=1}^N\frac{\partial\phi(\nu_s^M(\bold{f}))}{\partial
x_n}J_s^{M,i}(\bold{f})\right]\lambda_s^i\right|ds\right]\\
&+\mathbb{E}\left[\int_t^u\left|\sum\limits_{i=1}^M\left[\sum\limits_{n=1}^N\frac{\partial\phi(\nu_s^M(\bold{f}))}{\partial
x_n}J_s^{M,i}(\bold{f})\right]\lambda_s^i-
\sum\limits_{i=1}^M\left[\sum\limits_{n=1}^N\frac{\partial\phi(\nu_s^M(\bold{f}))}{\partial
x_n}\tilde J_s^{M,i}(\bold{f})\right]\lambda_s^i \right|ds\right]\\
&+\mathbb{E}\left[\int_t^u\left|\sum\limits_{i=1}^M\left[\sum\limits_{n=1}^N\frac{\partial\phi(\nu_s^M(\bold{f}))}{\partial
x_n}\tilde
J_s^{M,i}(\bold{f})\right]\lambda_s^i-\sum\limits_{i=1}^M\left[\sum\limits_{n=1}^N\frac{\partial\phi(\nu_s^M(\bold{f}))}{\partial
x_n}\tilde J_s^{M,i}(\bold{f})\right]\bar\lambda_s\right|ds\right].
\end{align}
Applying a Taylor expansion to $f\in C^\infty(\mathcal{O})$ and using the boundedness of its derivatives and the definition $c^i =\hat c^i/M$ , we find
\begin{align}\label{eq:boundJ}
J_t^{M,i}(f)\simeq \tilde J_t^{M,i}(f),
\end{align}
 where $a^M \simeq b^M $ means $\lim\limits_{M\rightarrow\infty}|a^M-b^M|=0$ and
\begin{align}
\tilde J_t^{M,i}(f):=\frac{1}{M}\frac{\partial f(X_t^i)}{\partial x}c^i.
\end{align}
Similarly, using the Taylor expansion of $\phi\in C^\infty(\mathbb{R}^N)$ we have  %Could this hold because in the jump term J we have a 1/M term so as M goes to infinity all second, third etc order terms go to zero much much faster?
\begin{align}\label{eq:boundPhi}
\phi(\nu_t^M(\bold{f})+J_t^{M,i}(\bold{f}))-\phi(\nu_t^M(\bold{f}))\simeq
\sum\limits_{n=1}^N\frac{\partial\phi(\nu_t^M(\bold{f}))}{\partial x_n} J_t^{M,i}(\bold{f}).
\end{align}

Using the finiteness of $\lambda_t^i$ from Proposition \ref{nonexpl}, equations \eqref{eq:boundPhi} and \eqref{eq:boundJ}, the boundedness of the derivatives of $f\in C^\infty(\mathcal{O})$ by their supremum, i.e. $||f||=\sup\limits_{(p,x)\in\mathcal{O}} |f(p,x)|$ and the bounds on the intensity given in \eqref{eq:delattre} we have that %\blu{(triple check this!)}
\begin{align}
\lim\limits_{M\rightarrow\infty}\mathbb{E}\left[\int_t^u|\mathcal{J}_s^{M}-\mathcal{\tilde
J}_s^M|ds\right] =0.
\end{align}
Similarly we have %\blu{(Why?! This part is important since it makes the whole expected value in limiting process valid!!)} Simply due to the 1/M^2 term and the fact that the derivatives are all bounded by the supremum. In case with additional X_t we need to check that X_t is also finite!
\begin{align}
\lim\limits_{M\rightarrow\infty}\mathbb{E}\left[\int_t^u|\mathcal{D}_s^M|ds\right]=0.
\end{align}
Define the operator $\mathcal{A}$ acting on the function $\Phi(\nu)$ defined in \eqref{eq:phi}, as
\begin{align}\label{eq:generator}
\mathcal{A}\Phi(\nu):=\sum\limits_{n=1}^N\frac{\partial\phi(\nu_t^M(\bold{f}))}{\partial
f_n}\left(\nu_t^M(\mathcal{L}^1f_n)\nu_t^M(I)-\nu_t^M(\mathcal{L}^2f_n)+\nu_t^M(\mathcal{L}^3f_n)+\nu_t^M(\mathcal{L}^4f_n)\right),
\end{align}
where $\mathcal{L}^4:=c\bar\lambda_t\partial_x.$ Then 
%\begin{align}
%\Phi(\nu_{t_{m+1}}^M)-\Phi(\nu_{t_m}^M)-\int_{t_m}^{t_{m+1}}\mathcal{A}\Phi(\nu_u^M)du,
%\end{align}
%is a martingale with initial mean zero and 
we have the following result:
\begin{lemma}[Limiting martingale problem]\label{lemma1} For any $\Phi\in\mathcal{S}$ and $0\leq t_1\leq ...\leq t_{m+1}\leq\infty$, with $m\in\mathbb{N}$ and $\Psi_j\in L^{\infty}(S)$ we have that $\mathcal{A}$ is the generator of the limiting martingale problem, i.e.
\begin{align}\label{eq:limmart}
\lim\limits_{M\rightarrow\infty}\mathbb{E}\left[\left(\Phi(\nu_{t_{m+1}}^M)-\Phi(\nu_{t_m}^M)-\int_{t_m}^{t_{m+1}}\mathcal{A}\Phi(\nu_u^M)du\right)\prod_{j=1}^m\Psi_j(\nu_{t_j}^M)\right]=0.
\end{align}
\end{lemma}

\subsection{Limiting process}\label{sec43}
Given the limiting martingale problem \eqref{eq:limmart} and assuming the existence and uniqueness of a limit point, we want to find the limiting process $\nu_t$ that satisfies equation \eqref{eq:limmart}.  Let $\bold{p}=(p^*,x)$. Define the following measure-valued process by %\blu{(I think we can easily include a distribution over initial condition here)}
\begin{align}\label{eq:nu}
\nu_t(A):=\mathbb{P}(X_t(\bold{p})\in A),
\end{align}
where $A\in \mathcal{B}(\mathbb{R})$ and the underlying limiting state process
$X(\bold{p})=(X_t(\bold{p}))_{t\geq 0}$ is a diffusion with time-varying coefficients given by
\begin{align}\label{eq:xxx}
 X_t(\bold{p})=x+\int_0^t\left(a\left(Q_1(s)-X_s(\bold{p})\right)+c\bar\lambda_s\right)ds+\sigma \int_0^tdW_s, \qquad t\geq 0,
\end{align}
with $\bar\lambda_t$ is defined in \eqref{eq:delambda} and
\begin{align}\label{eq:Qexpr}
 Q_1(t)=x+c\int_0^t\bar\lambda_sds.
\end{align}
Notice that $Q_1(t)$ satisfies the integral equation %{\bf [Do we use this?]}--> Yes, in Lemma 3.9 in [5]
\begin{align}\label{eq:Qdefine}
  Q_1(t)=e^{-at}\left(x+\int_0^t e^{as}\left(aQ_1(s)+c\bar\lambda_s\right)ds\right).
\end{align}
Using the definition of $\nu$ in \eqref{eq:nu} we have that
\begin{align}
\nu_t(I) = \int_\mathcal{O}x\nu_t(dx) = \mathbb{E}\left[X_t(\bold{p})\right],
\end{align}
where the underlying state process $X_{t}(\bold{p})$ is given by \eqref{eq:xxx}.
Notice that
\begin{align}
\mathbb{E}\left[X_t(\bold{p})\right]=e^{-at}\left(x+\int_0^te^{as}(aQ_1(s)+c\bar \lambda_s)\right)ds,
\end{align}
from which it follows that
\begin{align}\label{eq:defnuq}
Q_1(t)=\nu_t(I),
\end{align}
 where $I(x)= x$. We now prove that $\delta_\nu$ indeed satisfies the martingale problem in Lemma \ref{lemma1}:
\begin{theorem}[Limiting process]\label{theoremLimPro}
The empirical measure-valued process $\nu^M$ admits the weak convergence $\nu^M\rightarrow \nu$, as $M\rightarrow \infty$, where $\nu$ is defined as in \eqref{eq:nu}. 
%The empirical measure-valued process $\nu^M$ admits the weak convergence $\nu^M\rightarrow \nu$, as $M\rightarrow \infty$, where $\nu$ is defined as in \eqref{eq:nu}. 
Furthermore, $\nu^M(I)\rightarrow Q_1$.
\end{theorem}
\begin{proof}
Using the standard analysis of weak convergence as in Chapter 3 of \citet{ethier86}, the weak convergence $\nu^M\rightarrow \nu$ as $M\rightarrow \infty$ follows from Lemma \ref{lemma1} and Lemmas \ref{rc1}, \ref{rc2} and uniqueness of the limit point. 
In other words, if we define $\mathbb{Q}^M:=\mathbb{P}(\nu^M\in \mathcal{B}(D_S[0,\infty)))$, %note that D_S[0,\infty) is the Skorokhod space of right continuous functions from [0,\infty) to S
we have that $\mathbb{Q}^M$
converges to the solution $\mathbb{Q}$ of the martingale problem generated by $\mathcal{A}$ in
\eqref{eq:generator}. Next we show that the we have $\mathbb{Q}=\delta_\nu$, i.e. the limit
measure-valued process $\nu$ can indeed be represented as in \eqref{eq:nu}.  We have for $f\in
C^\infty(\mathcal{O})$ using the definition in \eqref{eq:nu1f} that
\begin{align}\label{eq:defnu}
\nu_t(f)=\mathbb{E}[f(X_t(\bold{p}))].
\end{align}
On the other hand, from \eqref{eq:xxx} and using It\^o's lemma, we have
\begin{align}
f(X_t(\bold{p}))=&f(x)+\int_0^t\frac{\partial f}{\partial
x}(X_s(\bold{p}))(aQ_1(s)-aX_s(\bold{p})+c\bar\lambda_s)ds+\frac{\sigma^2}{2}\int_0^t\frac{\partial
^2f}{\partial x^2}(X_s(\bold{p}))ds\\ &+\sigma\int_0^t\frac{\partial f}{\partial
x}(X_s(\bold{p}))dW_s.
\end{align}
Then recalling the definition of the operators $\mathcal{L}^*$ and the equality $Q_1(t)=\nu_t(I)$ from \eqref{eq:defnuq} we have %\blu{(check if ok to take lambda in expected value! Yes because it is deterministic right?! only time dependent)}
\begin{align}
\frac{\partial}{\partial t}\mathbb{E}[f(X_t(\bold{p}))]=&\frac{1}{2}\mathbb{E}\left[\sigma^2\partial_{xx}f(X_t(\bold{p}))\right]+Q_1(t)\mathbb{E}[a\partial_xf(X_t(\bold{p}))]+\mathbb{E}[c\bar\lambda_t\partial_xf(X_t(\bold{p}))]\\
&-\mathbb{E}[aX_s(\bold{p})\partial_xf(X_t(\bold{p}))]\\
=&\mathbb{E}[\mathcal{L}^3f(X_t(\bold{p}))]+\nu_t(I)\mathbb{E}[\mathcal{L}^1f(X_t(\bold{p}))]+\mathbb{E}[\mathcal{L}^4f(X_t(\bold{p}))]-\mathbb{E}[\mathcal{L}^2f(X_t(\bold{p}))].
\end{align}
% Note: expected value is taken over the Ito integral, and this is zero! Then after expeced value is taken we take the derivatives
So that, using \eqref{eq:defnu} we find
\begin{align}
\frac{d\Phi(\nu_t)}{dt}&=\sum\limits_{n=1}^N\frac{\partial\phi}{\partial
x_n}(\nu_t(\bold{f}))\frac{d\nu_t(f_n)}{dt}\\ &=\sum\limits_{n=1}^N\frac{\partial\phi}{\partial
x_n}\left(\nu_t(\mathcal{L}^3f)+\nu_t(\mathcal{L}^1f)\nu_t(I)+\nu_t(\mathcal{L}^4f)-\nu_t(\mathcal{L}^2f)\right)\\
&=\mathcal{A}\Phi(\nu_t).
\end{align}
So that for all functions $\Phi(\cdot)$ of the form \eqref{eq:phi} we have
\begin{align}
\Phi(\nu_t)=\Phi(\nu_s)+\int_s^t\mathcal{A}\Phi(\nu_u)du,\qquad 0\leq s<t<\infty,
\end{align}
and hence $\delta_\nu$ satisfies the martingale problem generated by $\mathcal{A}$.
\end{proof}
In other words, the propagation of chaos result from Theorem \ref{theoremLimPro} tells us that the empirical mean $\nu^M$ converges to a measure $\nu$ whose underlying process $X_t(\bold{p})$ reflects the Hawkes process through a \emph{time-dependent} drift.

\subsection{Extensions of the model}
In this section we shortly present results for several possible extensions of results presented in Section \ref{sec4}. In particular we derive the limiting empirical distribution when including a compound Hawkes process in the monetary reserve model considered in \eqref{eq:meanfieldsde}; a systematic risk factor, where the derivation is based on the result from \citet{giesecke15}; and furthermore prove a central limit theorem based on \citet{spiliopoulos14} which quantifies the fluctuation of the empirical distribution around its large system limit. 

\subsubsection{Compound Hawkes process}
If we include a compound Hawkes process in the initial log-monetary reserve SDE, i.e.
\begin{align}
dX_t^i = \frac{a^i}{M}\sum\limits_{k=1}^M(X_t^k-X_t^i)dt + \sigma^idW_t^i + c^idS_t^i,
\end{align}
where
\begin{align}
S_t^i = \sum\limits_{j=1}^{N_t^i}Z_j^i,
\end{align}
where $Z$ is an i.i.d. random variable with distribution function $F$, independent of $N_t^i$ and
$W_t^i$, such that $\lim\limits_{M\rightarrow\infty}\frac{1}{M} \sum\limits_{i=1}^M
\delta_{Z_\cdot^i} = y$. Then the limiting process is given by
\begin{align}
 X_t(\bold{p})=x+\int_0^t\left(a\left(Q_1(s)-X_s(\bold{p})\right)+cy\bar\lambda_s\right)ds+\sigma \int_0^tdW_s, \qquad t\geq 0.
\end{align}

\subsubsection{Systematic risk factor exposure}
Similar to the analysis of \citet{giesecke15} we can show that considering a non-vanishing systematic risk factor common to all the nodes in the system, we obtain a non-deterministic limiting behavior. Let $\mathcal{V}_t =\sigma(V_s,0\leq s\leq t)$ and $\mathcal{F}_t=\sigma((V_s,N_s^i,W_s^i),0\leq s\leq t, i\in\mathbb{N})$. Consider the following model for the log-monetary reserves
\begin{align}
&dX_t^i = a^i(\bar X_t-X_t^i)dt + \sigma^i dW_t^i + c^i dN_t^i + \beta^i dY_t,\label{eq:systrisksde}\\
&dY_t = b_0(Y_t)dt + \sigma_0(Y_t)dV_t,\;\; Y_0=y_0,
\end{align}
where $V_t$ is a standard Brownian motion independent of $W_t^i$ and $N_t^i$. In other words, $W_t^i$ represents a source of risk which is idiosyncratic to a specific name, while $Y_t$ is a systematic risk factor driven by a Brownian motion that is common to all the nodes in the network with the parameter $\beta^i$ representing the sensitivity of node $i$ to the $Y$. The systematic risk factor causes correlated changes in the monetary reserve process and thus acts as an additional source of clustering. As usual assume $p^i:=(a^i,\sigma^i,c^i,\beta^i)\rightarrow p^*:=(a,\sigma,c,\beta)$. Following the derivation in \citet{giesecke15} and defining $\Phi(y,\nu) = \phi_1(y)\phi_2(\nu(\bold{f}))$, and applying It\^ o's lemma as in the derivations for the original model we obtain for $0\leq t<u$
 \begin{align}
 \Phi(Y_u,\nu_u^M) = \Phi(Y_t,\nu_t^M) + \int_t^u(\phi_1(Y_s)\mathcal{C}_s^M+\phi_1(Y_s)\mathcal{D}_s^M+\phi_1(Y_s)\mathcal{J}_s^M + \mathcal{B}^{M,1}_s)ds + \int_t^u\mathcal{B}^{M,2}dV_s + \mathcal{M}_u-\mathcal{M}_t,
 \end{align}
 where we have defined
 \begin{align}
 \mathcal{B}^{M,1}_t := &\phi_1(Y_t)\sum_{n=1}^N\frac{\partial \phi_2(\nu_t^M(\bold{f}))}{\partial f_n}\nu_t^M(\mathcal{L}^5_{Y_t}f_n)+\phi_2(\nu(\bold{f}))\left(b_0(Y_t)\partial_y\phi_1(Y_t)+\frac{1}{2}\sigma_0^2(Y_t)\partial_{yy}\phi_1(Y_t)\right)\\
 &+\partial_y \phi_1(Y_t)\sum_{n=1}^N\frac{\partial \phi_2(\nu(\bold{f}))}{\partial f_n}\sigma_0(y)\nu_t^M(\mathcal{L}^6_{Y_t}f_n)\\
 \mathcal{B}^{M,2}_t:=&\phi_1(Y_t)\sum_{n=1}^N\frac{\partial \phi_2(\nu_t^M(\bold{f}))}{\partial f_n}\nu_t^M(\mathcal{L}^6_{Y_t}f_n)+\sigma_0(Y_t)\partial_y \phi_1(Y_t)\phi_2(\nu(\bold{f})),
 \end{align}
 with $\mathcal{L}^5_yf(p,x):=\beta^ib_0(y)\partial_x f(p,x) + \frac{1}{2}(\beta^i)^2\sigma_0^2(y)\partial_x f(p,x)$ and $\mathcal{L}^6_yf(p,x):=\beta^i\sigma_0(y)\partial_xf(p,x)$. Taking the limit of $M\rightarrow\infty$, using the limits derived in Section \ref{sec41} and the vanishing of the martingale in the limit (see also Lemma 7.2 in \cite{giesecke15}) and defining 
 \begin{align}
 \nu_t(f) = \mathbb{E}[f(X_t(\bold{p})|\mathcal{V}_t],
 \end{align}
  with 
  \begin{align}
  X_t(\bold{p})=x+\int_0^t\left(a\left(\nu_t(I)-X_s(\bold{p})\right)+c\bar\lambda_s\right)ds+\sigma \int_0^tdW_s+\beta\int_0^tdY_s, 
  \end{align}
  we obtain for the limiting process $\nu_t$ the following SPDE %volgens mij hebben al die \phi_1(y) dingen geen effect op de limieten, is in hun paper ook niet zo, 1/M^2 gaat nog steeds naar 0 in D en in J gaat die jump zoals altijd naar de constant poisson ding; verder hebben we hier geen martingale term omdat \nu defined is als de expected value
 \begin{align}
 d\nu_t(f(X_t)) = &\left(\nu_t(\mathcal{L}^1f(X_t))\nu_t(I)-\nu_t(\mathcal{L}^2f(X_t))+\nu_t(\mathcal{L}^3f(X_t))+\nu_t(\mathcal{L}^4f(X_t))+\nu_t(\mathcal{L}^5_{Y_t}f(X_t))\right)dt \\
 &+ \nu_t(\mathcal{L}^6_{Y_t}f(X_t))dV_t,
 \end{align}
 where we use Lemma B.1 and B.2 in \citet{giesecke15} to show that $\mathbb{E}\left[\int_0^tX_sdV_s|\mathcal{V}_t\right]=\int_0^t\mathbb{E}[X_s|\mathcal{V}_s]dV_s$.
The systematic risk factor thus does not vanish in the limit, and results in the stochastic partial differential equation for the limiting process of the empirical measure, instead of the deterministic behavior in the original model. 

\subsubsection{A Central Limit Theorem result}
Consider again the model defined in \eqref{eq:meanfieldsde}. In order to improve the first-order approximation of $\nu_t^M$ given in \eqref{eq:nu}, we can analyze the fluctuations of $\nu^M$ around its large system limit $\nu$. Following \citet{spiliopoulos14} define 
\begin{align}
\Xi_t^M=\sqrt{M}(\nu_t^M-\nu_t).
\end{align}
The signed-measure-valued process $\Xi^M$ weakly converges to the fluctuation limit $\bar\Xi$ in an appropriate space (in particular the convergence is considered in weighted Sobolev spaces in which the sequence $\Xi^M$, $M\in\mathbb{N}$ can be shown to be relatively compact; for discussion on this space, as well as the existence and uniqueness of the limiting point, we refer to Sections 7,8 and 9 in \citet{spiliopoulos14}). We start by deriving an expression for $\Xi_t^M$. Some terms in this expression will vanish in the limit of $M\rightarrow\infty$, and using the tightness of the processes (see Section 8 in \citet{spiliopoulos14}) and continuity of the operators in the expression for $\Xi^M$ we can pass to the limit and find the expression that the limiting fluctuation process satisfies. 

Subtracting $\nu_t$ from $\nu_t^M$ we find
\begin{align}
d\Xi_t^M(f) =& \left(\nu_t^M(\mathcal{L}^1f)\Xi_t^M(I)+\nu_t(I)\Xi_t^M(\mathcal{L}^1f)-\Xi_t^M(\mathcal{L}^2f) +\Xi_t^M(\mathcal{L}^3f)+\Xi_t^M(\mathcal{L}^4f)\right)dt+d\mathcal{M}^M_t(f)\\
&+\sqrt{M}\frac{1}{M}\sum_{i=1}^M(f(X_t^i+c^i)-f(X_t^i))d\tilde N_t -\sqrt{M}\frac{1}{M}\sum_{i=1}^Mc^i\frac{\partial f}{\partial x}\tilde N_t^i\\
&+ \sqrt{M}\left(\frac{1}{M}\sum_{i=1}^M(f(X_t^i+c^i)-f(X_t^i))\lambda_t^i-\nu_t^M(\mathcal{L}^4f)\right)dt,
\end{align}
where the martingale term is defined as
\begin{align}
\mathcal{M}^M_t(f) = \sqrt{M}\left(\frac{1}{M}\sum_{i=1}^M\int_0^t\sigma^i\partial_x f dW_s^i+\int_0^t\frac{1}{M}\sum_{i=1}^Mc^i\frac{\partial f}{\partial x}d\tilde N_s^i\right).
\end{align}
Using the limiting expressions for the Hawkes jump term and a Taylor approximation from Section \ref{sec41}, we have
\begin{align}
&\sqrt{M}\bigg|\frac{1}{M}\sum_{i=1}^M(f(X_t^i+c^i)-f(X_t^i))-\frac{1}{M}\sum_{i=1}^Mc^i\frac{\partial f}{\partial x}\bigg| \leq \frac{K^2}{M\sqrt{M}}\norm{\frac{\partial^2f}{\partial x^2}}.\label{eq:limitingonjump}
%&\lim_{M\rightarrow\infty}\sup_{0\leq t\leq T}\mathbb{E}\left[\int_0^t\sqrt{M}\left(\frac{1}{M}\sum_{i=1}^M(f(X_t^i+c)-f(X_t^i)) -\nu_t^M(\mathcal{L}^4)\right)d\tilde N_t\right]^2=0
\end{align}
Thus one can show by taking the limit $M\rightarrow\infty$, using \eqref{eq:limitingonjump} and Assumption \ref{ass1} that the sequence $\{\Xi_t^M,t\in[0,T]\}_{M\in\mathbb{N}}$ converges in distribution to the limit point $\{\Xi_t\in[0,T]\}$ that satisfies 
\begin{align}
\Xi_t(f) = \Xi_0(f) + \int_0^t\left(\nu_s^M(\mathcal{L}^1f)\Xi_s(I)+\nu_s(I)\Xi_s(\mathcal{L}^1f)-\Xi_s(\mathcal{L}^2f) +\Xi_s(\mathcal{L}^3f)+\Xi_s(\mathcal{L}^4f)\right)ds + \mathcal{M}_t(f),
\end{align}
where $\{\mathcal{M}_t,t\in[0,T]\}$ is the distribution valued, continuous square integrable martingale with a deterministic quadratic variation to which the sequence $\{\mathcal{M}^M_t,t\in[0,T]\}_{M\in\mathbb{N}}$ converges in distribution (note: unlike in the LLN cases, the martingale term does not vanish in the CLT scaling case). By a martingale CLT (see 7.1.4 in \citet{ethier86}) $\mathcal{M}$ is Gaussian. %Not Gaussian right?! I mean we have the Hawkes jump term... How could this be Gaussian...
This implies the following second-order approximation $\nu_t^M\overset{d}\approx \nu_t + \frac{1}{\sqrt{M}}\Xi_t$, giving a more accurate approximation for finite banking systems.  %why dont we use martingale problem for \Xi as usual? we do use the martingale problem, since we derive d\Xi^M; only for some reason we don't need to define the separating space etc... Then why do the the limiting terms hold in expectation, and where do we use expectation for \Xi equation? Because the convergence is in distribution, i.e. $\mathbb{E}[\Xi^M] \rightarrow^d \mathbb{E}[\Xi]$? No, not only that, convergence in distribution states that martingale term might differ. Then why does martingale term not become zero? We just state that \Xi_t-\Xi_0-\int...dt is a martingale, this I agree with and totally satisfies what we had before (they even state that it is a centered Gaussian, so indeed zero mean), taking expectations would result in sth similar to our original martingale problem. BUT then why do the above limiting equations hold in expectation, but we still apparently use them WITHOUT taking expectations?

\section{Systemic risk in a large network}\label{sec5}
In this section we introduce several systemic risk indicators to quantify the risk in our network and to show the particular dependence of the risk on the underlying parameters. We first remark on the difference between the monetary reserve with a Hawkes process and one with an independent Poisson process:
\begin{remark}[Independent Poisson process versus Hawkes process] Consider an independent Poisson process with intensity $\mu$. It is straightforward to see that
\begin{align}
\bar \lambda_t :=\mu + \int_0^t\alpha e^{-\beta(t-s)}\bar\lambda_sds\geq \mu,
\end{align}
since we assume $\alpha, \beta\geq 0$.
 Therefore, for $c<0$ we have that $Q_1(t)\leq \tilde Q_1(t)$, with $Q_1$ and $\tilde Q_1$ being the averages from a Poisson jump with intensity $\lambda_t$ and a jump with intensity $\mu$ respectively. Thus, in the limit $M\rightarrow \infty$, using $\nu^M(I)\rightarrow Q_1(t)$, we have as expected that the Hawkes process increases the default risk in the network.
\end{remark}
\subsection{Risk indicators}
Here we show how one can measure the systemic risk in a large network using the limiting dynamics $X_t(\bold{p})$. We propose to compute systemic risk in the mean-field model based on the fraction of banks that
have transitioned from a normal to a defaulted state. We define the risk indicator as the expected
value of the fraction of banks that throughout time $t\in[0,T]$ have dropped below the default
level $D$,
\begin{align}
 \text{\rm SR}^M:=\frac{1}{M}\sum\limits_{i=1}^M\caratt_{\big\{\min\limits_{0\leq t\leq T}
 X_t^i\leq D\big\}}.
\end{align}
Note that from Theorem \ref{theoremLimPro} we have $\lim\limits_{M\rightarrow\infty}\nu_t^M=\nu_t$ for a continuous function $f$ of $X_t^i$. For the indicator function over $t\in [0,T]$ we consider the approximate relationship to hold 
\begin{align}
\lim\limits_{M\rightarrow\infty}\text{\rm SR}^M&\approx \mathbb{E}\left[\caratt_{\big\{\min\limits_{0\leq
t\leq T}X_t(\bold{p})\leq D\big\}}\right],
\end{align}
in which the average over the indicator function of the $M$ monetary reserve processes is thus replaced by the indicator of the limiting process.

Furthermore, similar to \citet{capponi15} we can define the average distance to default as
\begin{align}
\text{\rm ADD}^M(t):=\mathbb{E}\left[\frac{1}{M}\sum\limits_{i=1}^M X_t^i\right].
\end{align}
Note that $(\nu_t^M;M\in\mathbb{R})$ is uniformly integrable, i.e. for each $t\geq 0$
\begin{align}
\sup\limits_{M\in\mathbb{N}}\mathbb{E}\left[\left|\nu_t^M(I)\right|^2\right]<\infty,
\end{align}
the proof of which is similar to the proof of Lemma \ref{lemmabound} in Appendix \ref{app2} and Lemma B.2 in \citet{capponi15}.
Then for the average distance to default indicator we use the following limiting result
\begin{align}
\lim\limits_{M\rightarrow\infty} \text{\rm ADD}^M(t) = Q_1(t)
\end{align}
with $Q_1(t)$ as in \eqref{eq:defnuq}. 
Note that in the case of independent Poisson jumps with intensity $\lambda$, the limit of the
\text{\rm ADD} indicator is given by $\lim\limits_{M\rightarrow\infty}\text{\rm ADD}^M(t)=
x+c\lambda t$. This is in contrast to the case of the Hawkes jumps for which we have
$\lim\limits_{M\rightarrow\infty}\text{\rm ADD}^M(t) = x+c\int_0^t\bar\lambda_sds$.

\subsection{Numerical results}
We set $M = 300$, i.e. sufficiently large, and analyze how our approximation formulas for the various indicators of systemic risk compare to the corresponding
Monte-Carlo estimate. The latter is obtained by simulating $M$ interacting processes $X_t^i$, $i\in I_M$ using an Euler approximation of \eqref{eq:meanfieldsde}.
\begin{remark}[Computation of $\bar \lambda_t$] Define the partition of $[0,T]$ as $0=t_0<t_1<...<t_K=T$ with $\Delta t:=t_{i}-t_{i-1}$. Then we approximate the integral in \eqref{eq:delambda} as
\begin{align}
\bar\lambda_{t_{i+1}} \approx \bar\lambda_{t_{i}} + \Delta t g(\Delta t)\bar\lambda_{t_{i}},
\end{align}
and $\bar \lambda_0:=\mu$. Using the approximated $\bar\lambda_t$ we compute $Q_1(t)$ as
\begin{align}
Q_1(t_{i+1}) \approx Q_1(t_i) + \Delta t c\bar \lambda_{t_{i}},
\end{align}
where $Q_1(0) = x$.
\end{remark}

\begin{table}[H]
\caption{Monte Carlo estimates versus the LLN approximation for the systemic risk indicators with $\mu=0.01$, $\alpha = 1$, $\beta=1.2$, $a=0.5$, $\sigma=0.5$, $\hat c=-0.2$ and $D=0$.}\label{tab1}
\begin{center}
\begin{tabular}{c|c|c||c|c}
&\multicolumn{2}{c}{Monte Carlo}&\multicolumn{2}{c}{Approximation}\\\hline\hline $x_0$&$\text{\rm
SR}$&$\text{\rm ADD}(T)$&$\text{\rm SR}$&$\text{\rm ADD}(T)$\\\hline\hline
0.002&0.945&0.007&0.949&0.007\\ 0.1&0.821&0.096&0.816&0.096\\ 0.2&0.658&0.197&0.652&0.197\\
0.5&0.252&0.497&0.261&0.497\\ 0.8&0.057&0.797&0.058&0.797\\ 1&0.016&0.998&0.017&0.997\\
\end{tabular}
\end{center}
\end{table}
\begin{table}
\caption{Monte Carlo estimates versus the LLN approximation for the systemic risk indicators with $\mu=0.05$, $\alpha = 1$, $\beta=1.2$, $a=0.5$, $\sigma=0.5$, $\hat c=-0.2$ and $D=0$.}\label{tab2}
\begin{center}
\begin{tabular}{c|c|c||c|c}
&\multicolumn{2}{c}{Monte Carlo}&\multicolumn{2}{c}{Approximation}\\\hline\hline $x_0$&$\text{\rm
SR}$&$\text{\rm ADD}(T)$&$\text{\rm SR}$&$\text{\rm ADD}(T)$\\\hline\hline
0.01&0.947&-0.005&0.946&-0.007\\ 0.1&0.826&0.085&0.830&0.083\\ 0.2&0.669&0.186&0.653&0.183\\
0.5&0.262&0.486&0.269&0.483\\ 0.8&0.061&0.785&0.061&0.783\\ 1&0.017&0.985&0.016&0.0.983\\
\end{tabular}
\end{center}
\end{table}

In Table \ref{tab1} and \ref{tab2} we present the results for our approximation and the
Monte-Carlo estimates for $5000$ simulations, $100$ time steps, $T=1$ and $M=300$. As expected the
systemic risk in the network, as quantified by both $\text{\rm SR}$ and $\text{\rm ADD}$,
decreases as the initial monetary reserve value increases. Furthermore, a higher mean jump intensity
$\mu$ results in a less stable network. In Figure \ref{fig11} we show the LLN estimates for the
systemic risk and the average distance to default for the Hawkes and Poisson process for different values of the initial reserve $x_0$. Our claims
of the Hawkes process adding an additional default risk in the model are verified also in these
numerical results, as the systemic risk indicator for the Hawkes process is considerably larger,
while the average monetary reserves are consistently lower than for an independent Poisson
process. Therefore, the self- and cross-exciting shock modelled through the Hawkes process is an
additional form of contagion in the network, resulting in the network being more prone to a
systemic risk event.

\begin{figure}[H]
\begin{center}
   \caption{LLN estimates for the systemic risk (L) and LLN estimates for the average distance to default (R) at time $T=1$ with $\mu=0.2$, $\alpha = 1.2$, $\beta=1.2$, $a=0.5$, $\sigma=0.5$, $c=-1$ and $D=0$ for a independent Poisson process and the Hawkes process for $x_0\in [0,1]$}
\includegraphics[scale = 0.4]{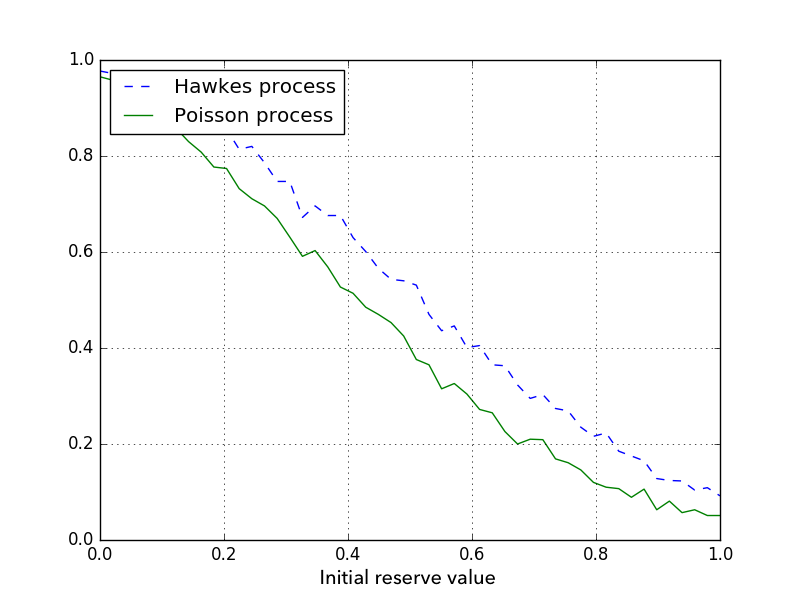}
\includegraphics[scale = 0.4]{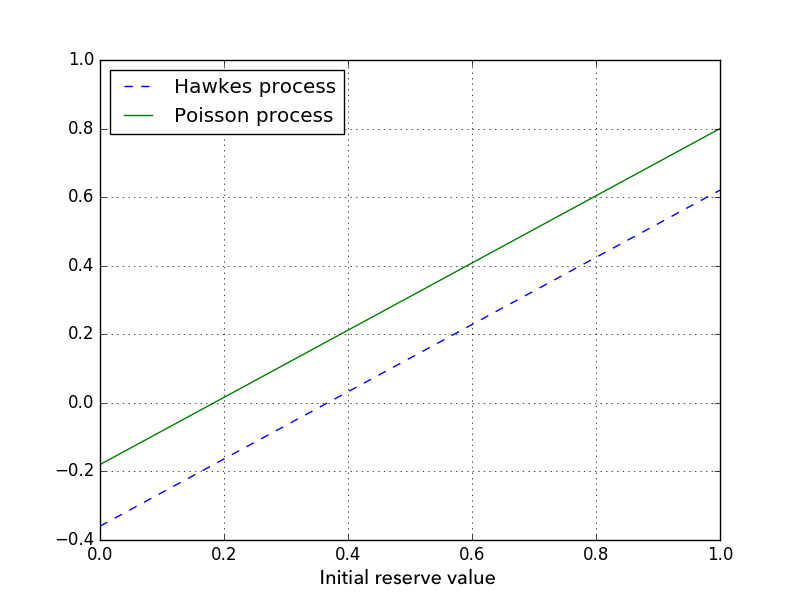}
\end{center}
\label{fig11}
\end{figure}

\subsubsection{Calibrating the model}
Calibration of the model considered in \eqref{eq:meanfieldsde} with heterogeneous coefficients, in particular for the large banking system, is a complex task. In \citet{ait15} the authors considered a calibration for a Hawkes diffusion model used to model asset returns and developed method of moments estimates for the parameters of the model. Even after making simplifying assumptions on the intensity, the model was fitted only on pairs of assets. The calibration of the mean-field SDE with Hawkes jumps for a large number of banks is therefore besides the scope of this paper and left for further research. However, the limiting expression derived in Section \ref{sec43} can be used to derive a simple and efficient way of calibrating the model. In particular, we can calibrate the average distance to default given by $Q_1(t)$ in \eqref{eq:Qexpr} by fitting it to an average of a \emph{sufficiently large} number of assets, resulting in the calibrated parameters $x$, $c$, $\mu$, $\alpha$ and $\beta$. In particular, consider the asset price as a proxy for the monetary reserve process and consider the average of the components of the S\&P500 index over the period of 2008-07-14 until 2008-10-21. Calibrating the deterministic expression for $Q_1(t)$ to the actual average distance to default we obtain the following set of parameters: $\mu = 0.3$, $x = 1300$, $\alpha = 0.07$, $\beta = 0.11$ and $c=-1.6$. It can be argued that the the assumption of regularity of the parameters in the limit (see Assumption \ref{ass1}) is too strong and disenables calibrating to actual excitation. Nevertheless, using this simple and efficient way of calibrating the model, we see from the left-hand side of Figure \ref{figcalib} that contagion is sufficiently captured; in particular note that the Poisson process is unable to model the necessary contagion as seen from the right-hand side of Figure \ref{figcalib}, while the SDE with the Hawkes process provides a much better fit. 

\begin{figure}[H]
\begin{center}
   \caption{Calibrated model for $Q_1$ on the S\&P500 data showing excitation effects (L) and the average of 5000 simulated SDE paths of $X_t(\textbf{p})$ (R)}
\includegraphics[scale = 0.5]{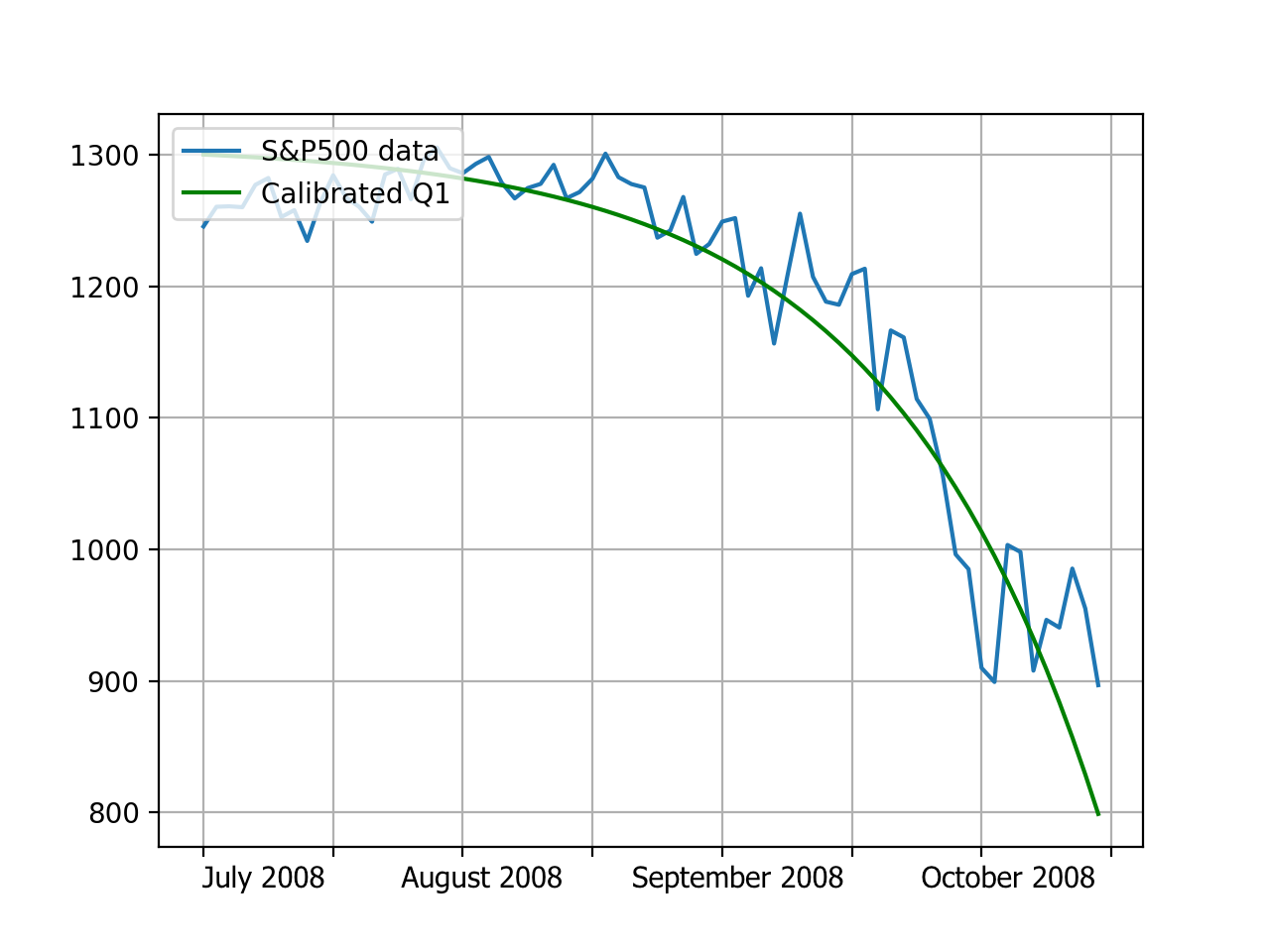}
\includegraphics[scale = 0.5]{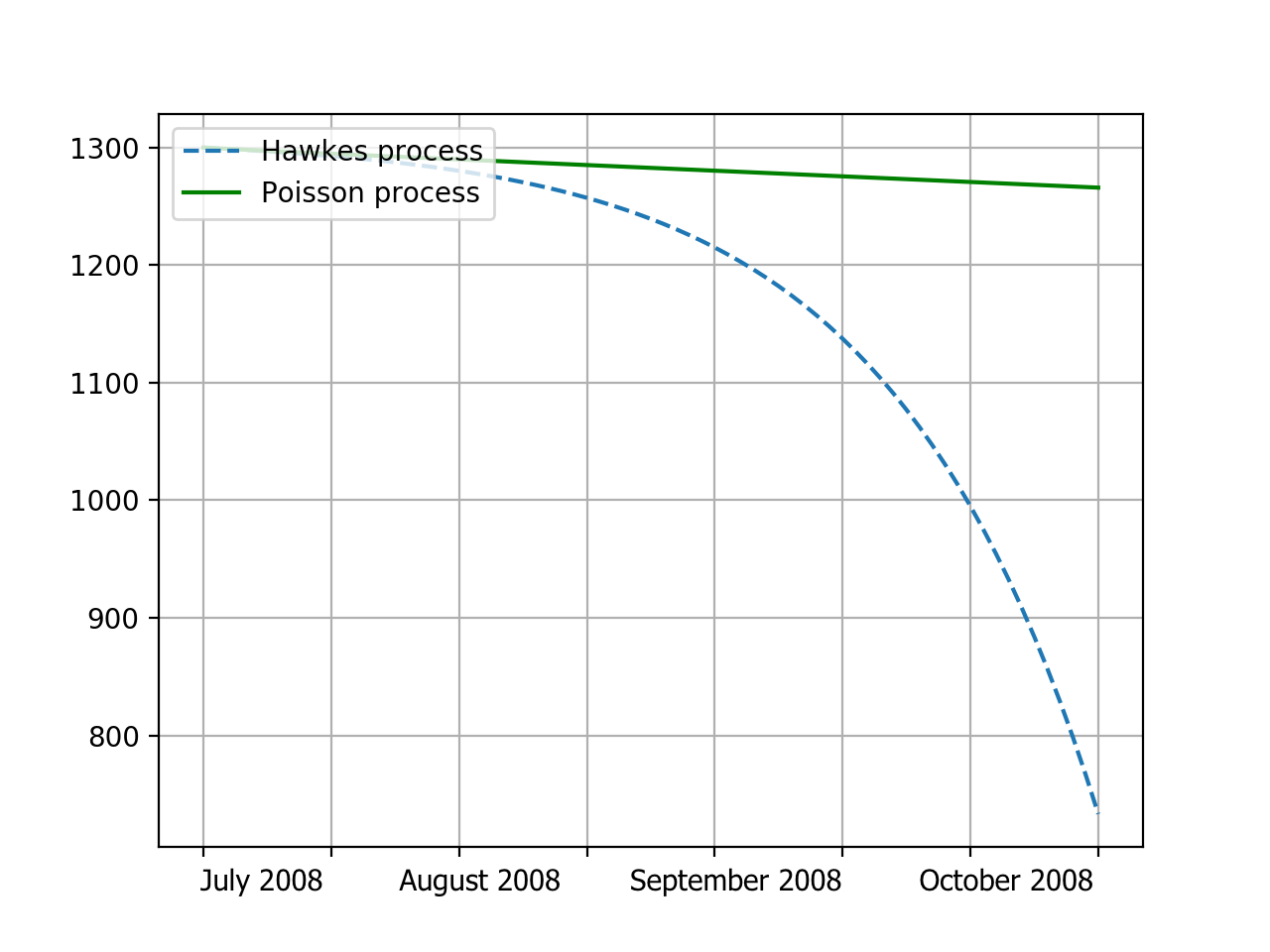}
\label{figcalib}
\end{center}
\end{figure}

\section{Conclusion}
In this paper we have studied the effects of considering an additional self-exciting and clustering shock that impacts the monetary reserve or asset value of the nodes of the interbank system. The nodes are assumed to interact through the drift, and additionally are subjected to a Hawkes-distributed shock. In this way the jump activity varies over time resulting in jump clustering and the shocks propagate through the network in a contagious manner. This allows us to model both default propagation due to interbank loans as well as propagation due to linked balance sheets and financial acceleration. We started with a numerical analysis of the interbank model in which we showed that the Hawkes jumps results in a non-negligible tail-probability of multiple defaults occuring at the same time. We then considered the effects of the Hawkes process in a mean-field interaction model for the monetary reserve process and derived a weak convergence of the empirical mean to a measure whose underlying process reflects the Hawkes process through a time-dependent drift term. Finally we defined several risk indicators and their LLN approximations which can be used for quantifying the risk in large systems and showed that the LLN estimates perform accurately compared to Monte-Carlo simulations. We conclude that the clustering Hawkes jumps result in an additional and important source of default propagation in the network and should not be ignored.

\section*{Acknowledgements}
This research is supported by the European Union in the the context of the H2020 EU Marie Curie Initial Training Network project named WAKEUPCALL.

\appendix

\section{Proofs}\label{app2}
The next Lemma is a boundedness result of the moment estimate of the log-monetary reserve process.
\begin{lemma} \label{lemmabound} For $n=1,2$ %\blu{(do we even need higher $n$'s..? I think no and for 1 and 2 its easy)}
and $T\geq 0$ we have
\begin{align}
\sup\limits_{0\leq t\leq
T,\;M\in\mathbb{N}}\frac{1}{M}\sum\limits_{i=1}^M\mathbb{E}\left[\left|X_t^i\right|^n\right]<+\infty.
\end{align}
\end{lemma}
\begin{proof}
Let $n\in\{1,2\}$. Recall the constant $C_p$ bounding the parameters $(p^i,X_0^i)$ from assumption \ref{ass1}. From It\^o's formula we have
\begin{align}
\mathbb{E}\left[|X_t^i|^n\right] =& \mathbb{E}\left[|X_0^i|^n\right] +
a^i\mathbb{E}\left[\int_0^tn|X_s^i|^{n-1}(\bar X_s-X^i_s)ds\right] +
\frac{1}{2}(\sigma^i)^2\mathbb{E}\left[\int_0^tn(n-1)|X_s^i|^{n-2}\right]\\
&+\sigma^i\mathbb{E}\left[\int_0^tn|X_s^i|^{n-1}dW_s^i\right] +
\mathbb{E}\left[\int_0^t\left[|X_s^i+c^i|^n-|X_s^i|^n\right]dN_s^i\right].
\end{align}
Using Young's inequality we have
\begin{align}
a^in{X_t^i}^{n-1}\bar X_t-a^in|X_t^i|^{n-1}X_t^i&\leq
a^i\frac{n}{M}\sum\limits_{k=1}^M|X_t^i|^{n-1}|X_t^k|-a^in|X_t^i|^{n}\\ &\leq
C_p\frac{1}{M}\sum\limits_{k=1}^M|X_t^k|^n+(2n-1)C_p|X_t^i|^n.
\end{align}
Applying Young's inequality twice yields
\begin{align}
\frac{n(n-1)}{2}(\sigma^i)2|X_t^i|^{n-2}&\leq \frac{n(n-1)}{2}\left(\frac{n-2}{n-1}|X_t^i|^{n-1}+\frac{1}{n-1}(\sigma^i)^{2n}\right)\\
&\leq \frac{n(n-1)}{2}\left(\frac{n-2}{n}|X_t^i|^{n}+\frac{1}{n}+\frac{1}{n-1}C_p\right).
\end{align}
Finally, using Young's inequality and Proposition \ref{nonexpl} there exists a constant $C_n$ independent of $M$ such that %\blu{(Check, because not super precise here... Also we need bound on $\lambda^2$...)}
\begin{align}
\mathbb{E}\left[\int_0^t\left[|X_s^i+c^i|^n-|X_s^i|^n\right]dN_s^i\right]&=\mathbb{E}\left[\int_0^t\left[|X_s^i+c^i|^n-|X_s^i|^n\right]\lambda_s^ids\right]\\
&\leq \frac{1}{2}\mathbb{E}\left[\int_0^t|c^iX_s^i|^{2(n-1)}+|c^i|^{2n}ds\right]+\frac{1}{2}\mathbb{E}\left[\int_0^t(\lambda_s^i)^2ds\right]\\
&\leq C_n(1+\mathbb{E}\left[\int_0^t|X_s^i|^nds\right].
\end{align}
The statement then follows from applying Gronwall's Lemma and the fact that the limiting constants are independent of $M$.
\end{proof}

In order to conclude weak convergence of the empirical measure $\nu_t^M$ to $\nu_t$ we need to determine the limiting martingale problem (as done in Section \ref{sec41}), show uniqueness of the limit point and its existence (i.e. tightness of the sequence of measure-valued processes). We provide here a sketch of the proof for the latter. We have to prove that the sequence of measure-valued processes $\{\nu^M\}_{M\in\mathbb{N}}$ defined by \eqref{eq:nuemp} are relatively compact when viewed as a sequence of random processes on the Skorokhod space $D_S([0,\infty])$, the collection of c\`adl\`ag functions from $[0,\infty)$ to $S$. This is necessary to ensure that the laws of $\nu^M$ have at least one limit point  (see also Chapter 2 and 3 of \citet{ethier86}). The complication arising from using a Hawkes process is the feedback loop in the intensity, however due to Theorem \ref{lemmaint} we know that the intensity is bounded and thus the system will not explode. The relative compactness will be implied by the following two Lemmas: Lemma \ref{rc1} on compact containment and Lemma \ref{rc2} on the regularity of the $\nu^M$'s.
% Compact containment: for each m\in mathbb{N} there will be a set K such that \nu_t^M will belong to that set for t\in [0,T] with high probablity
% Regularity: \nu_t^M - \nu_s^M is bounded by a function (t-s)
% In Giesecke this all follows also from Lemma 3.4 which basically bounds the intensity. In our case we have the non-explosive case (see Laub) when \alpha/\beta < 1, and since we take \alpha=1/M\alpha, this should hold. So in our case the non-epxlosiveness is determined in proposition 3.4, and maybe it will be good to mention additionally this requirement, or check the relation between this requirement and spectral radius thing, because I believe it is similar since the spectral radius in our case is something like \alpha/\beta. Using then the fact that the intensity is non-explosive (where explosive means N_t-N_s = \infty for (t-s)<\infty) we have bounds similar to Capponi and Giesecke... 
\begin{lemma}\label{rc1}
For every $T>0$ and any smooth function $f\in C^\infty(\mathcal{O})$, we have
\begin{align}
\lim_{m\to\infty}\sup\limits_{M\in\mathbb{N}}\mathbb{P}\left(\sup\limits_{0\leq t\leq
T}|\nu_t^M(f)|\geq m\right)=0.
\end{align}
\end{lemma}
\begin{proof}
From \eqref{eq:nuito} we have the following decomposition
\begin{align}\label{eq:nudecomp}
\nu^M_t(f) = \nu_0^M(f) + A_t^M + B_t^M + C_t^M + D_t^M,
\end{align}
where we have defined
\begin{align}\label{eq:defops}
A_t^M&:=\frac{1}{M}\int_0^t\sum\limits_{i=1}^Ma^i\partial_x f(X_s^i)(\nu_s^M(I)-X_s^i)ds,\\
B_t^M&:=\frac{1}{2M}\int_0^t\sum\limits_{i=1}^M(\sigma^i)^2\partial_{xx}f(X_s^i)ds,\\ C_t^M&:=
\frac{1}{M}\int_0^t\sum\limits_{i=1}^M\left(\sigma^i\partial_xf(X_s^i)dW_s^i\right),\\
D_t^M&:=\int_0^t\left[\frac{1}{M}\sum\limits_{i=1}^M(f(X_s^i+c^i)-f(X_{s-}^i))\right]dN_s^i.
\end{align}
Then we need to bound $\mathbb{E}\left[\sup\limits_{0\leq t\leq T}|(\cdot)_t^M|\right]$ for each
of the terms defined above. Denote for $f\in C^\infty(\mathcal{O})$ the supremum norm with
$||f||=\sup\limits_{(p,x)\in\mathcal{O}}|f(p,x)|$. We will use the dominating constant $C_p$ from
assumption \ref{ass1}. For $A_t^M$, $B_t^M$, $C_t^M$ the estimates are similar to
\citet{capponi15} and we omit the details here and just give the estimates
\begin{align}
&\mathbb{E}\left[\sup\limits_{0\leq t\leq T}|A_t^M|\right]\leq C_p\left\|\frac{\partial
f}{\partial x}\right\|\int_0^T\frac{1}{M}\sum\limits_{i=1}^M\mathbb{E}\left[|X_s^i|^2\right]ds +
C_p\left\|\frac{\partial f}{\partial x}\right\|,\\ &\mathbb{E}\left[\sup\limits_{0\leq t\leq
T}|B_t^M|\right]\leq \frac{C_p}{2}\left\|\frac{\partial^2 f}{\partial x^2}\right\|T,\\
&\mathbb{E}\left[\sup\limits_{0\leq t\leq T}|C_t^M|\right]\leq C_TC_p\left\|\frac{\partial
f}{\partial x}\right\|(T+1).
\end{align}
Then we have by the mean-value theorem and using Proposition \ref{nonexpl} which implies the
existence of a constant $C_\lambda$ such that $\mathbb{E}[\lambda^i_t]<C_\lambda$ that
\begin{align}
\mathbb{E}\left[\sup\limits_{0\leq t\leq T}|D_t^M|\right]&\leq
\sum\limits_{i=1}^M\mathbb{E}\left[\int_0^T\frac{1}{M}|f(X_s^i+c^i)-f(X_{s-}^i)|dN_s^i\right]\\
&\leq \left\|\frac{\partial f}{\partial
x}\right\|\frac{1}{M}\sum\limits_{i=1}^Mc^i\int_0^T\mathbb{E}[\lambda_s^i]ds\\ &\leq
\left\|\frac{\partial f}{\partial x}\right\| C_pC_\lambda T.
\end{align}
Using Lemma \ref{lemmabound}, we can find a positive constant $C$ such that
\begin{align}
\sup\limits_{M\in\mathbb{N}}\mathbb{E}\left[\sup\limits_{0\leq t\leq
T}\left|\nu_t^M(f)\right|\right]<C.
\end{align}
\end{proof}

Define $\mathbb{E}_t[\cdot]:=\mathbb{E}[\cdot|\mathcal{F}_t]$.
\begin{lemma}\label{rc2}
Let $h(x,y)=|x-y|\wedge 1$ for any $x,y\in\mathbb{E}$. Then there exists a positive random
variable $H_M(\gamma)$ with $\lim\limits_{\gamma\rightarrow
0}\sup\limits_{M\in\mathbb{N}}\mathbb{E}[H_M(\gamma)]=0$ such that for all $0\leq t\leq T$, $0\leq
u\leq \gamma$ and $0\leq v\leq \gamma\wedge 1$, we have
\begin{align}
\mathbb{E}_t\left[h^2(\nu_{t+u}^M(f),\nu_t^M(f))h^2(\nu_t^M(f),\nu_{t-v}^M(f)\right]\leq \mathbb{E}_t[H_M(\gamma)],
\end{align}
where the function $f\in C^{\infty}(\mathcal{O})$.
\end{lemma}
\begin{proof}
We have from \eqref{eq:nudecomp}
\begin{align}
(\nu_{t+u}^M-\nu_t^M)(f) = A_{t+u}^M-A_t^M+B_{t+u}^M-B_t^M+C_{t+u}^M-C_t^M + \mathcal{M}_{t+u}^M-\mathcal{M}_t^M+P_{t+u}^M-P_t,
\end{align}
where $A_t^M$, $B_t^M$, $C_t^M$ are defined in \eqref{eq:defops} and
\begin{align}
&\mathcal{M}_t^M:=\int_0^t\left[\frac{1}{M}\sum\limits_{i=1}^M(f(X_s^i+c^i)-f(X_{s-}^i))\right]d\tilde
N_s^i,\\
&P_t^M:=\int_0^t\left[\frac{1}{M}\sum\limits_{i=1}^M(f(X_s^i+c^i)-f(X_{s-}^i))\right]\lambda_s^ids,
\end{align}
where we have used the fact that the compensated counting process $\tilde N_t^i:= N_t^i-\int_0^t\lambda^i_sds$ is a $\mathcal{F}_t$-local martingale. We have
\begin{align}
h^2\left(\nu_{t+u}^M(f),\nu_t^M(f)\right)\leq& 16\big[\left|A_{t+u}^M-A_t^M\right|^2+\left|B_{t+u}^M-B_t^M\right|^2+\left|C_{t+u}^M-C_t^M\right|^2\\
&+\left|\mathcal{M}_{t+u}^M-\mathcal{M}_t^M\right|^2+\left|P_{t+u}^M-P_t^M\right|^2\big].
\end{align}
Let $0\leq u\leq \gamma$. For the bounds on the first three differences we refer to Lemma 3.5 in \citet{capponi15}. For the fourth difference, using the martingale property and It\^o Isometry for the martingale $(\mathcal{M}_t^M)$ with quadratic variation $[\tilde N_t,\tilde N_t]=N_t$, %\blu{(CHECK, see also Bacry)}
 the mean-value theorem, Assumption \ref{ass1} and Proposition \ref{nonexpl}, and the bound (6.1) in \citet{giesecke13} we find
\begin{align}
\mathbb{E}_t\left[\left|\mathcal{M}_{t+u}^M-\mathcal{M}_t^M\right|^2\right]&=\mathbb{E}_t\left[\left|\mathcal{M}_{t+u}^M\right|^2-\left|\mathcal{M}_t^M\right|^2\right]\\
&=\sum\limits_{i=1}^M\mathbb{E}_t\left[\int_t^{t+u}\frac{1}{M}\left|f(X_s^i+c^i)-f(X_{s-}^i)\right|^2dN_s^i\right]\\
&=\sum\limits_{i=1}^M\mathbb{E}_t\left[\int_t^{t+u}\frac{1}{M}\left|f(X_s^i+c^i)-f(X_{s-}^i)\right|^2\lambda^i_sds\right]\\
&\leq C_p\left\|\frac{\partial f}{\partial
x}\right\|^2\frac{1}{M}\sum\limits_{i=1}^M\mathbb{E}_t\left[\int_t^{t+u}\lambda^i_sdt\right]\\
&\leq C_p\frac{1}{2}\left\|\frac{\partial f}{\partial
x}\right\|^2\gamma^{\frac{1}{4}}\frac{1}{M}\sum\limits_{i=1}^M\mathbb{E}\left[1+\int_0^{T}(\lambda^i_s)^2dt\right].
\end{align}
With the mean-value theorem and Assumption \ref{ass1} we find
\begin{align}
\left|P_{t+u}^M-P_t^M\right|&=\left|\sum\limits_{i=1}^M\int_t^{t+u}\left[\frac{1}{M}(f(X_s^i+c^i)-f(X_{s-}^i))\right]\lambda_s^ids\right|\\
&\leq C_p\left\|\frac{\partial f}{\partial
x}\right\|\frac{1}{M}\sum\limits_{i=1}^M\int_t^{t+u}|\lambda^i_s|ds\\ &\leq
C_p\frac{1}{2}\left\|\frac{\partial f}{\partial
x}\right\|\gamma^{\frac{1}{4}}\frac{1}{M}\sum\limits_{i=1}^M\left(1+\int_0^{T}(\lambda^i_s)^2dt\right).
\end{align}
Then using Lemma \ref{lemmabound} and Proposition \ref{nonexpl} we can finish the proof.  %note we do not need lemma b.2 as in capponi since we dont have the nu(I) part in the martingale difference
\end{proof}

Then if uniqueness of the limit point $\nu_t$ holds (see e.g. the proof of Lemma C.1 in \cite{capponi15}), we can thus conclude that the sequence $\nu_t^M$ converges weakly to the limit point $\nu_t$ and we thus conclude that weak convergence holds.

\bibliographystyle{plainnat}
\bibliography{biblio}
\end{document}